\documentclass[12pt,letterpaper]{article}
\usepackage{amssymb,amsmath}
\usepackage{mathtools}
\usepackage{amsthm}
\usepackage{bbm}
\usepackage{graphicx}
\usepackage{enumerate}
\usepackage{caption}
\usepackage{natbib}
\usepackage{url} 
\usepackage{bm}
\usepackage{color}
\usepackage{placeins} 
\usepackage{silence} 
\newtheorem{proposition}{Proposition}
\newtheorem{corollary}{Corollary}

\theoremstyle{definition}

\DeclareMathOperator*{\median}{\mbox{median}}

\DeclareMathOperator{\diag}{\mbox{diag}}
\DeclareMathOperator*{\argmin}{\mbox{argmin}}

\DeclareMathOperator{\tr}{\mbox{trace}}

\newcommand{\bzero}{\boldsymbol 0}
\newcommand{\bone}{\boldsymbol 1}

\newcommand{\btop}{\boldsymbol{\top}}

\newcommand{\br}{\boldsymbol r}
\newcommand{\bc}{\boldsymbol c}
\newcommand{\bv}{\boldsymbol v}
\newcommand{\bx}{\boldsymbol x}

\newcommand{\btx}{\boldsymbol{\widetilde{x}}}
\newcommand{\by}{\boldsymbol y}
\newcommand{\bty}{\boldsymbol{\widetilde{y}}}
\newcommand{\bz}{\boldsymbol z}

\newcommand{\bA}{\boldsymbol A}
\newcommand{\bB}{\boldsymbol B}
\newcommand{\bC}{\boldsymbol C}
\newcommand{\bD}{\boldsymbol D}
\newcommand{\bI}{\boldsymbol I}
\newcommand{\bP}{\boldsymbol P}
\newcommand{\bR}{\boldsymbol R}
\newcommand{\bS}{\boldsymbol S}
\newcommand{\bhS}{\boldsymbol{\widehat{S}}}
\newcommand{\bT}{\boldsymbol T}
\newcommand{\bhT}{\boldsymbol{\widehat{T}}}
\newcommand{\bU}{\boldsymbol U}
\newcommand{\bhU}{\boldsymbol{\widehat{U}}}
\newcommand{\bV}{\boldsymbol V}
\newcommand{\bhV}{\boldsymbol{\widehat{V}}}
\newcommand{\bW}{\boldsymbol W}
\newcommand{\bX}{\boldsymbol X}
\newcommand{\bhX}{\boldsymbol{\widehat{X}}}
\newcommand{\bY}{\boldsymbol Y}
\newcommand{\bZ}{\boldsymbol Z}

\newcommand{\btX}{\boldsymbol{\widetilde{X}}}
\newcommand{\btY}{\boldsymbol{\widetilde{Y}}}

\newcommand{\bbeta}{\boldsymbol \beta}
\newcommand{\bgamma}{\boldsymbol \gamma}
\newcommand{\bhgamma}{\boldsymbol{\widehat{\gamma}}}
\newcommand{\bGamma}{\boldsymbol \Gamma}
\newcommand{\bhGamma}{\boldsymbol{\widehat{\Gamma}}}
\newcommand{\bdelta}{\boldsymbol \delta}
\newcommand{\bDelta}{\boldsymbol \Delta}

\newcommand{\bTheta}{\boldsymbol \Theta}
\newcommand{\bhTheta}{\boldsymbol{\hat{\Theta}}}

\newcommand{\bhbeta}{\boldsymbol{\hat{\beta}}}

\newcommand{\bmu}{\boldsymbol \mu}
\newcommand{\hmu}{\hat{\mu}}
\newcommand{\bhmu}{\boldsymbol{\hat{\mu}}}
\newcommand{\bSigma}{\boldsymbol \Sigma}
\newcommand{\bhSigma}{\boldsymbol{\widehat{\Sigma}}}

\newcommand{\halpha}{\hat{\alpha}}
\newcommand{\hbeta}{\hat{\beta}}

\newcommand{\hsigma}{\hat{\sigma}}

\begin{document}

\def\spacingset#1{\renewcommand{\baselinestretch}
{#1}\small\normalsize} \spacingset{1}


\title{\bf Challenges of cellwise outliers}	

\author{Jakob Raymaekers\\
  {\normalsize Department of Quantitative 
	Economics, Maastricht University, 
	The Netherlands}\\ \\
	Peter J. Rousseeuw\\
	{\normalsize Section of Statistics and
	Data Science,  
	University of Leuven, Belgium}\\ \\}
\date{February 4, 2023}
\maketitle

\begin{abstract}
It is well-known that real data often 
contain outliers. The term outlier 
typically refers to a case, that is, 
a row of the $n \times d$ data matrix.
In recent times a different type has 
come into focus, the cellwise outliers.
These are suspicious cells (entries) 
that can occur anywhere in the data
matrix. 
Even a relatively small proportion of 
outlying cells can contaminate over half 
the rows, which is a problem for rowwise 
robust methods. 
This article discusses the challenges
posed by cellwise outliers, and some
methods developed so far to deal with them.
New results are obtained on cellwise 
breakdown values for location, covariance 
and regression. 
A cellwise robust method is proposed 
for correspondence analysis, with real 
data illustrations.
The paper concludes by formulating some 
points for debate.
\end{abstract}

\vspace{0.5cm} 


\spacingset{1.1} 

\section{Introduction} 
\label{sec:intro}

Real-world data often contains elements
that do not follow the pattern 
suggested by the majority of the data. 
Such outliers may be gross errors with
the potential to severely distort a
statistical analysis, but they may also
be valuable pieces of information that 
warrant further inspection.
Therefore we want the ability to detect 
outliers, no matter what caused them.

The approach of robust statistics, 
as formalized by e.g. \cite{Huber1964} and 
\cite{hampel1986robust} has been to first
obtain a robust fit, i.e. a model that fits 
the majority of the data well, and then to
look for cases that deviate from that fit. 
An important assumption underlying this 
approach is that of casewise contamination,
in which a case is either an outlier or 
free of any contamination. 
This is also called rowwise contamination,
because cases are typically encoded as
rows of the data matrix.
Different formalizations of casewise
contamination exist, but the most common 
one is to assume that the observed data 
was generated from 
a clean distribution $F$ with probability 
$1-\varepsilon > 0.5$ and from an arbitrary 
distribution $H$ with probability 
$\varepsilon < 0.5$. 
In other words, we observe $\bX$ defined as
\begin{equation}\label{eq:thcm}
\bX = (1 - B)\bY + B\bZ
\end{equation}
where $B \sim \mbox{Bernoulli}(\varepsilon)$
takes on the values 0 or 1,  $\bY \sim F$,
$\bZ \sim H$, and $B$ is independent of
$\bY$ and $\bZ$. 
This is known as the {\it Tukey-Huber 
contamination model} (THCM).
A commonly used equivalent formulation is 
$\bX \sim F_{\varepsilon} \coloneqq 
 (1-\varepsilon)F + \varepsilon H$. 
The general goal is to estimate the 
characteristics of $F$ given that we only
observe $F_{\varepsilon}$ and, importantly, 
we don't assume anything about $H$.

With this model in mind, a wide variety of 
robust statistical methods has been 
developed for settings including the
estimation of location and covariance, 
fitting regression models, carrying out
principal component analysis and 
discriminant analysis, analyzing 
time series and so on, see for instance
the books by \cite{hampel1986robust},
\cite{RL1987}, and \cite{maronna2019robust}. 

Under the THCM model it is assumed that 
a case is either coming from an arbitrary
distribution $H$ that may have no relation
to $F$, or is a perfect draw from $F$.
Therefore, methods developed under this
model tend to either trust all coordinates
of a case, or downweight all of them
simultaneously.
But there are also situations where one
or more coordinates (entries, cells) of
a case are suspect whereas the others
are fine.
By downweighting such a case entirely
we may lose valuable information residing
in its uncontaminated cells. 

This explains why nowadays also another
contamination model is being considered,
that of cellwise outliers.
It was first published by
\cite{alqallaf2009}.
They assumed that we observe a random 
variable $\bX$ given by 
\begin{equation*}
\bX = (\bI - \bB) \bY + \bB\bZ
\end{equation*}
where $\bB$ is a diagonal matrix, 
whose diagonal entries can only take 
the values zero or one.
This allows some components of the 
observed random vector $\bX$ to be 
contaminated, whereas others are clean. 
The model thus assumes that 
the data were initially generated by 
the intended distribution $F$, but 
afterward contaminated by replacing 
some cells by other values. 
In a given row it may happen that no
cells are replaced, or a few, or even
all of them.
Figure~\ref{fig:cellrow} illustrates the 
difference between the casewise (rowwise)
model in the left panel and the cellwise 
model in the right panel. 

\begin{figure}[!ht]
\centering
\includegraphics[width=0.7\columnwidth]
   {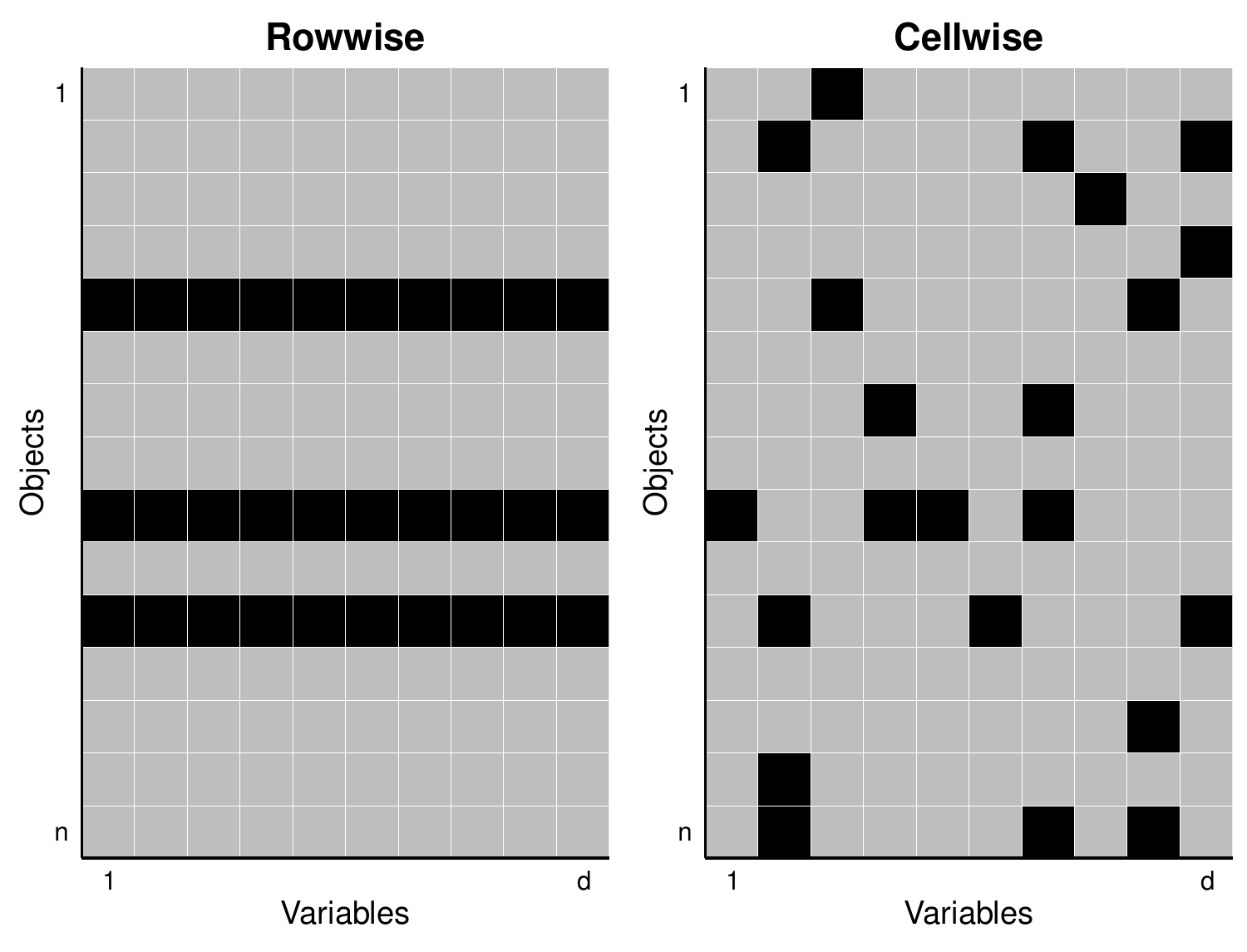}
\caption{Rowwise outlier model (left) 
  versus cellwise outlier model (right). 
  Black means outlying. A small fraction 
  of cellwise outliers can contaminate 
  many rows.}
\label{fig:cellrow}
\end{figure}

Cellwise outliers embody a paradigm shift 
from rowwise outliers. Arguably 
the most drastic change is the fact that a 
small percentage of contaminated cells can 
contaminate a large fraction of the rows.
For instance, assume that cells are
contaminated with probability $\varepsilon$
and that this happens for all cells 
independently. In that situation the 
probability that a case (row) is 
contaminated in at least one of its cells
equals
\begin{equation}\label{eq:probcont}
  P[\mbox{row is contaminated}] = 
	1 - (1-\varepsilon)^d
\end{equation}
which grows very quickly with the dimension 
$d$. For example, in $d = 15$ dimensions
a fraction $\varepsilon = 0.05$ of 
contaminated cells suffices to have on 
average over $50\%$ of contaminated rows, 
and then the methods developed for the 
THCM model are no longer reliable.

Long before the notion of cellwise 
outliers appeared in print in
\cite{alqallaf2009}, it was 
discussed informally by some. 
One of us remembers clearly that the 
topic was brought up in the early 
eighties in seminar talks by Alfio 
Marazzi and Werner Stahel at ETH 
Z\"urich. It was presented as an open
problem, including 
formula~\eqref{eq:probcont}.
The general feeling then was that the 
tools available at that time were 
insufficient to address this challenge.

The remainder of our paper is organized 
as follows. In section~\ref{sec:meth} we 
give an overview of the methodology that 
has been developed for dealing with 
cellwise outliers, in various settings.
In section~\ref{sec:bdv} we prove some
new results on cellwise breakdown values.
Next, section~\ref{sec:CA} proposes the 
first cellwise robust method for 
correspondence analysis with two real
data illustrations.
Section~\ref{sec:discussion} then 
evaluates the current state of  
research on cellwise robustness and 
raises some points for discussion.
Section~\ref{sec:conc} concludes.

\section{Methods for dealing with 
         cellwise outliers}
\label{sec:meth}

In this section we give an overview of 
methodology developed for dealing 
with cellwise outliers. 
We first treat ``point clouds'', i.e. 
single-class multivariate numerical data. 
Then we discuss other models including 
regression, classification, principal 
component analysis, and clustering. 

\subsection{Detecting cellwise outliers}

Detecting cellwise outliers turns out to 
be a difficult task. 
The most straightforward way of trying 
to detect cellwise outliers is to look 
for outliers in each variable separately. 
A simple approach is to robustly 
standardize each variable yielding a 
robust z-score, and then to compare these 
z-scores to quantiles of a reference 
distribution (e.g. the standard Gaussian) 
to decide whether or not to flag a cell.
A more involved version of detecting
marginal outliers is to use a univariate 
filter such as the one proposed by 
\cite{gervini2002class}. Either way, 
to allow for a wider range of potentially 
asymmetric distributions one could first 
robustly transform the variables to 
central normality as in  
\cite{raymaekers2021transforming}.

While flagging marginal outliers is 
fast, only rather extreme cellwise 
outliers will be identified reliably,
since relations between the variables 
are not taken into account.
For instance, consider the bivariate
setting in 
Figure~\ref{fig:bivariateoutliers}. 
Case 2 has a marginally outlying
value $x_{2,2}$\,, whereas case 3 has a 
marginally outlying $x_{3,1}$\,. 
But for case 4 things are not so clear. 
The marginal approach does not see 
anything wrong with its $x_{4,1}$ or 
$x_{4,2}$ even though it is clearly 
an outlier relative to the pattern of 
the large majority of the datapoints,
for which $X_1$ and $X_2$ are 
negatively related.
This illustrates that flagging
marginal outliers is insufficient
in general.

\begin{figure}[!ht]
\centering
\includegraphics[width=0.5\columnwidth]
  {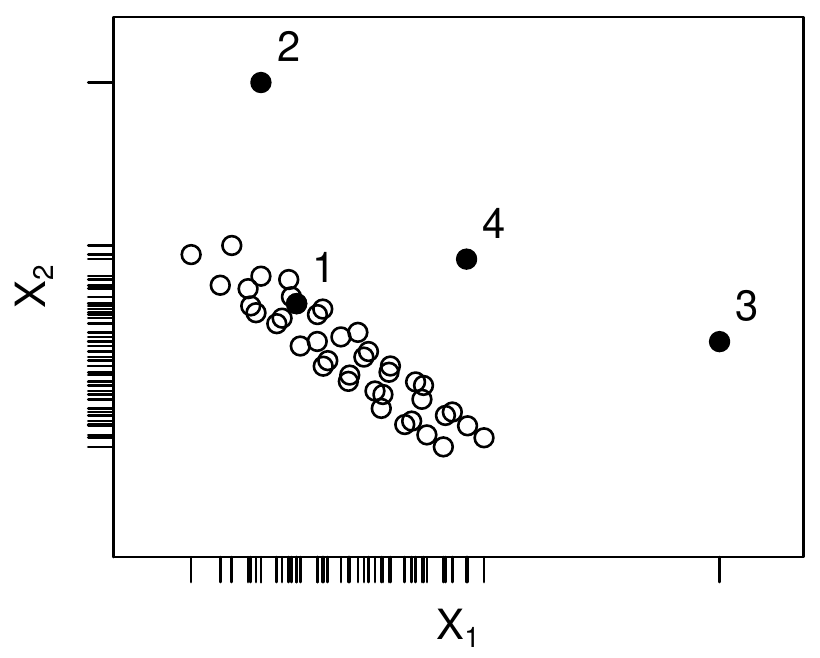}
\caption{Cellwise outliers need not be 
         marginally outlying.}
\label{fig:bivariateoutliers}
\end{figure}

Moreover, based on 
Figure~\ref{fig:bivariateoutliers}
alone it is undecidable whether case 4
has a cellwise outlier in $X_1$ or 
$X_2$ or maybe both.
But note that it might be possible to
decide whether $x_{4, 1}$ or $x_{4,2}$ 
is a cellwise outlier if we had more 
variables in the data that were 
correlated with $X_1$ and $X_2$\,. 
This suggests that in this particular
situation the curse of dimensionality 
may be something of a blessing. 

This idea motivates approaches for 
detecting cellwise outliers that 
consider more than just the marginal 
distributions. One direction is to 
extend the univariate GY filter of 
\cite{gervini2002class}. 
\cite{Leung2017} extended it to a 
bivariate filter in the context of 
estimating location and covariance. 
More recently, 
\cite{saraceno2021robust2} went to 
more than two dimensions. 

Another approach is the Detect 
Deviating Cells (\textbf{DDC}) 
algorithm of \cite{Rousseeuw2018}. 
This procedure uses robust simple 
linear regressions of each variable 
on other variables in the data.
Combining this information yields
predicted values for each cell. 
The difference between the original
dataset and its prediction 
provides a residual for each cell.
One can then flag cellwise outliers 
by large standardized cellwise
residuals.
This approach has received some 
attention, and in fact the most 
successful versions of the filter
approach incorporate \textbf{DDC} in 
their filter.
\cite{walach2020cellwise} used a 
similar approach for metabolomics
data by aggregating outlyingness over 
pairwise log-ratios.

More recently, some approaches were
developed that try to go beyond 
combining bivariate relations by
incorporating all relations in the data.
\cite{cellHandler} consider
each case in turn and look for a
vector parameter that captures the shift 
needed to bring the case into the fold. 
Through $L^1$ regularization they 
enforce sparsity on this shift, so that
only a minority of the cells of the case
are actually shifted. This algorithm is 
called \textbf{cellHandler}, and 
was shown to be effective in identifying 
complex adversarial cellwise outliers.
The inputs of \textbf{cellHandler} 
are the data matrix as well as one of 
the robust estimates of its covariance 
matrix described in the next subsection.

\begin{table}[ht]
\centering
\vspace{2mm}
\begin{tabular}{rcccc}
\hline
d & Univariate GY & Bivariate GY & 
  DDC & cellHandler\\ 
\hline
5 & 0.0007 & 0.0027 & 0.0005 & 0.0076 \\ 
10 & 0.0012 & 0.0087 & 0.0012 & 0.0186 \\ 
20 & 0.0037 & 0.0304 & 0.0033 & 0.0674 \\ 
50 & 0.0059 & 0.1672 & 0.0150 & 0.5594 \\ 
\hline
\end{tabular}
\caption{Flagging cellwise outliers: 
  computation times in seconds for a 
  dataset with $n = 1000$ and varying 
  dimension $d$.}
\label{tab:flagging}
\end{table}

Table~\ref{tab:flagging} compares the
computation times of several methods for
detecting cellwise outliers. The sample
size was always $n=1000$ and the dimension 
$d$ ranged from 5 to 50. The data were 
generated from a Gaussian distribution
with mean zero and covariance matrix 
$\Sigma_{jk} = (-0.9)^{|j-k|}$ in which
10\% of the cells were replaced by the
value 5. All computations were carried
out in \textsf{R}, and the reported times 
are averages over 10 replications.
As expected the univariate filter is the
fastest. The bivariate GY filter takes
more time, also compared to \textbf{DDC}. 
The \textbf{cellHandler} method is the 
slowest in relative terms but still fast.

As there are several approaches for 
detecting cellwise outliers, some guidance
is useful to select one of them in a
given situation. When there are reasons
to assume that all cellwise outliers
are marginally outlying, a univariate 
filter is sufficient. When they 
are expected to stand out in bivariate 
plots one can e.g. use the bivariate
GY filter~\citep{Leung2017}.
However, for less computational effort 
one might as well run \textbf{DDC}, 
which combines 
information from many variables and 
provides predicted values for the
flagged cells. Also \textbf{cellHandler}
does that. In order to choose between
the latter pair, it is good to know that
\textbf{cellHandler} requires a robust 
covariance estimate which assumes $n > p$, 
and also assumes a roughly \mbox{elliptical}
data shape. \textbf{DDC} does not make
these assumptions and is faster, but when 
it is applicable \textbf{cellHandler}
is more accurate for detecting complex
outliers.

A somewhat related approach is the 
\textbf{SPADIMO} algorithm by 
\cite{Debruyne2019}. It is not purely
cellwise because it starts by identifying 
outlying cases by large robust distances
from a casewise robust location and
covariance matrix. Next, it tries to 
identify the cells responsible for the 
outlyingness of those cases.
For this it uses regularized regression 
to identify sparse directions of maximal 
outlyingness.

\subsection{Estimating location and covariance}

If there is any reassuring news 
in the study of cellwise outliers, 
it is probably that multivariate 
location can still be estimated 
reliably from a robustness 
viewpoint by combining robust 
location estimates on the marginals. 
The typical example is the 
coordinatewise median, given by
\begin{equation}\label{eq:comed}
   \bhmu(\bX) = 
   (\median_{i=1}^n x_{i1},\ldots,
   \median_{i=1}^n x_{id})
\end{equation}
for a dataset $\bX$ with $n$ rows and
$d$ columns.
While such a componentwise estimate of 
location may not always have the most 
desirable properties, it is at least 
consistent under general conditions and 
very robust. 
The same is true for the estimation of 
the scales of the variables, in other 
words, the diagonal of a covariance 
matrix.

However, when we leave the comfort zone 
of estimating multivariate location and 
componentwise scale, things become 
much more challenging. 
For cellwise robust estimation of a 
covariance matrix, \cite{VanAelst2011}
proposed the cellwise Stahel-Donoho 
estimator. 
This is an extension of the casewise 
Stahel-Donoho covariance estimator 
\citep{stahel1981robuste,
donoho1982breakdown}. 
It was seen to perform reasonably 
well under casewise contamination,
but in the setting of cellwise
contamination it is now outperformed
by more recent methods, as illustrated
in Figure 1 of \cite{Agostinelli2015}
and Figure 9 of \cite{Rousseeuw2018}. 

\cite{danilov2010robust} thoroughly 
studied the approach of filtering the 
data before estimating a covariance 
matrix. 
He considered the three-step procedure of 
(i) detecting the contaminated cells, 
(ii) processing them in some way to 
reduce their influence, and (iii) 
estimating the location and the 
covariance matrix. 
For the detection, he looked at marginal 
approaches based on robust z-scores, as 
well as multivariate approaches based on 
partial Mahalanobis distances. 
He ended up proposing the latter with the 
caveat that the computation time scales 
exponentially with the dimension $d$ 
if one wants to identify cellwise 
outliers in all possible subspaces. 
For processing the outlying cells he 
compared winsorizing, censoring, and 
setting them to missing. 
The final proposal used a multivariate 
cell detection method and then set the 
outlying cells to missing, followed by 
the MLE for the resulting incomplete 
dataset. Continuing in this 
direction, the two-step 
generalized S-estimator (2SGS) was 
proposed by \cite{Agostinelli2015} 
and extended by \cite{Leung2017}.
This was the best performer for a number 
of years. 
The 2SGS method combines a bivariate 
filter with the generalized S-estimator 
for missing data \citep{danilov2012}, 
and is implemented in the 
\textsf{R}-package \texttt{GSE}  
\citep{Leung2019}.

Recently there have been new attempts 
at estimating the covariance matrix 
under adversarial cellwise outliers. 
\cite{cellHandler} proposed 
an iterative approach of (i) running
\textbf{cellHandler} to flag outlying cells, 
(ii) setting these cells to missing,
and (iii) applying the EM algorithm 
to re-estimate the covariance matrix. 
While this approach performs quite well, 
it is purely algorithmic in nature 
without a tractable objective function,
and therefore its properties are 
difficult to analyze. 
However, building on these ideas  
\cite{cellMCD} proposed 
a cellwise robust version of the 
well-known casewise robust minimum 
covariance determinant (MCD) estimator 
of \cite{rousseeuw1984least}. 
The new cellMCD method does optimize an 
objective function and was proven to 
possess a high breakdown value 
under cellwise contamination. 
Its algorithm iterates a kind of
C-steps, which (like those for the
casewise MCD) always lower the
objective function. Therefore the
algorithm is guaranteed to converge
to a solution, which may be a local 
minimum.
The initial estimator uses a combination 
of the robust correlation estimator of
\cite{raymaekers2021fast} and \textbf{DDC},
but several initial estimates could
be used.
The methods described in this paragraph 
are implemented in the \textsf{R} package 
\texttt{cellWise} \citep{R:cellWise}.

\begin{table}[ht]
\centering
\vspace{2mm}
\begin{tabular}{rccccc}
\hline
d & Classical & MCD & S & 2SGS & cellMCD \\ 
\hline
5 & 0.0000 & 0.0423 & 0.1448 & 0.5730 & 0.2403 \\ 
10 & 0.0002 & 0.0950 & 0.3145 & 1.6476 & 0.3733 \\ 
20 & 0.0002 & 0.3065 & 1.0548 & 15.206 & 1.3153 \\ 
50 & 0.0015 & 2.0049 & 8.0532 & 213.02 & 24.770 \\ 
\hline
\end{tabular}
\caption{Estimation of location and 
   covariance: computation times in seconds 
   for a dataset with $n = 1000$ and 
   varying dimension $d$.}
\label{tab:cov}
\end{table}

Table~\ref{tab:cov} lists the 
computation times of the classical mean 
and covariance matrix, as well as those 
of the casewise robust MCD and 
S-estimators. These are compared to the
cellwise robust methods 2SGS and cellMCD.
The setup was the same as for 
Table~\ref{tab:flagging}, with all
methods run in \textsf{R} in their default 
version. We note that the cellwise 
robust methods take the most time, but
are still feasible. The 2SGS method was
the slowest, partly due to its default 
initial estimator that could be 
replaced by a different one.

\cite{cwMLE} recently proposed a 
method to statistically analyze data 
with cellwise weights, and implemented 
it for the estimation of location and 
covariance. 
One possible application is to first
assign weights to cells depending on
how outlying they are (e.g. based on
a standardized residual), and then
to compute a reweighted location and
covariance matrix by cellwise weighted
maximum likelihood. 

\subsection{Estimating a precision matrix}
 
In many applications it is the inverse 
of the covariance matrix, also known as 
the precision matrix, that is of main 
interest. 
In low dimensions, it is of course 
possible to estimate a covariance matrix 
and then to invert the estimate, but 
this is no longer the case in 
high dimensions because the estimated
covariance matrix is often singular. 
Rather than regularizing the covariance 
matrix, one can then estimate the 
precision matrix and regularize it
at the same time. 
This induces zeroes in the precision 
matrix, which correspond to partial 
correlations of zero which in turn 
correspond to conditional independence 
in the Gaussian graphical model. 
One of the most popular estimators of
a high-dimensional precision matrix 
is the graphical lasso of 
\cite{friedman2008sparse}. 
It requires the input of an estimate 
of the covariance matrix $\bhSigma$, 
usually the classical empirical 
covariance. 
The precision matrix $\bTheta$ is then 
estimated as
\begin{equation} \label{eq:GLASSO}  
  \bhTheta  = 
  \argmin_{\bTheta \succcurlyeq 0} 
	\left(\tr\left(\bhSigma \bTheta\right) 
	- \log(\det(\bTheta)) + 
	\lambda \sum_{j\neq \ell} 
    |\bTheta_{j\ell}|\right)\;.
\end{equation}
A natural question is whether a robust 
version of $\bhSigma$ would result in a 
robust precision matrix estimator, and 
this is indeed the case. 
This idea has been studied 
empirically and theoretically by 
\cite{ollerer2015robust}, 
\cite{croux2016robust}, 
\cite{tarr2016}, \cite{loh2018high}, 
and \cite{katayama2018robust}, using 
different $\bhSigma$ based on robust 
pairwise correlations.
The $\bhSigma$ in these studies 
are typically of the form 
\begin{equation}
 \bhSigma_{j\ell} = \mbox{s}\!\left(
  \bX_{.j}\right)\mbox{s}\!\left(
	\bX_{.\ell}\right) \mbox{r}\!\left(
	\bX_{.j}, \bX_{.\ell}\right)
\end{equation}
where $\mbox{s}(\cdot)$ denotes a robust 
scale estimator, often the $Q_n$ estimator 
of \cite{rousseeuw1993alternatives}, and 
$\mbox{r}(\cdot)$ denotes a robust 
correlation estimator such as Spearman's 
rank correlation. 
Note that $\bhSigma$ needs to be positive 
semidefinite for~\eqref{eq:GLASSO} to 
work. When $\bhSigma$ is not positive 
semidefinite it is first modified, e.g.
by adding a small multiple of the 
identity matrix. 
The resulting precision matrix 
estimators perform well under 
non-adversarial cellwise contamination, 
are fast to compute, and can be analyzed
theoretically due to the well-defined 
objective function~\eqref{eq:GLASSO}. 
But since they focus on pairwise 
correlations, their performance is 
likely to suffer under more adversarial 
cellwise outliers. 

\subsection{Regression}

Regression in the presence of cellwise 
outliers has mainly been studied in the 
linear model with numerical 
predictors and response. More precisely, 
we model the observed pairs 
$(\bx_1, y_1), \ldots, (\bx_n, y_n)$ 
where $\bx_i$ is a $p \times 1$ vector 
for $i = 1, \ldots, n$ as
\begin{equation}\label{eq:reg}
  y_i = \alpha + \bbeta \bx_i 
	      + \mbox{noise}_i
\end{equation}
where the regression coefficients $\alpha$ 
and $\bbeta=(\beta_1,\ldots,\beta_p)$ 
need to be estimated.

In casewise robust linear regression, 
an observed pair $(\bx_i, y_i)$ is 
considered to be either outlying or 
entirely clean. 
For cellwise robust regression the setup 
is different.
It remains true that the response $y_i$
is either outlying or clean because it 
is univariate, so it has only one cell.
The main difference lies in the 
contamination of the predictor variables. 
Instead of working with either outlying
or entirely clean 
$\bx_i = (x_{i1},\ldots,x_{ip})$, we now 
allow for some of the cells $x_{ij}$ to 
be contaminated while the others are 
clean. 
We would like to use the clean cells of
$\bx_i$ in the estimation of $\alpha$ 
and $\bbeta$, but at the same time limit 
the harmful influence of the bad cells
of $\bx_i$\,.

One of the earliest proposals for 
cellwise robust regression is the 
shooting S-estimator 
of \cite{ollerer2016shooting}. 
It uses the coordinate descent 
algorithm \citep{bezdek1987}, also 
called ``shooting algorithm''. 
If we fit only the $j$-th slope 
coefficient $\beta_j$ while keeping 
the other slopes fixed at their
previous values, we can do
\begin{equation} \label{eq:shooting}
  y_i^{(j)} \sim \;
  \halpha_j^{\mbox{\tiny{new}}} + 
  \hbeta_j^{\mbox{\tiny{new}}} x_{ij} 
  \;\;\;\;\;\mbox{ where }\;\;\;\;\;
  y_i^{(j)} \coloneqq y_i - 
  \bhbeta_{-j}^{\mbox{\tiny{old}}} 
  \bx_{i,-j}\;\;.
\end{equation}
This can be seen as a simple regression 
of a new response $\by^{(j)}$ on the 
$j$-th predictor. 
The shooting S-estimator makes two 
changes to the coordinate descent
algorithm. 
The first is that in the simple 
regression~\eqref{eq:shooting} the OLS 
regression is replaced by a casewise
robust S-estimator. 
The second is that the response 
$\by^{(j)}$ is `robustified' in each 
iteration, to make sure the cellwise 
outliers don't propagate to the new 
responses in the iteration steps. 
\cite{bottmer2022sparse} recently 
proposed a version of shooting S
regression for high dimensions which
uses hard thresholding, yielding 
sparse estimates of the regression 
coefficients.

An alternative proposal was made by 
\cite{leung2016robust}. Their model 
assumes joint multivariate normality
of the pairs $(\bx_i, y_i)$ before
the cells are contaminated.
They start by estimating the location 
$\bhmu$ and covariance matrix $\bhSigma$ 
of the joint distribution of the
pairs $(\bx_i, y_i)$ by the 2SGS 
estimator of \cite{Agostinelli2015},
with a modified filter in the first step. 
Then the slope coefficients and the
intercept are estimated from $\bhmu$
and $\bhSigma$ by the usual formulas
\begin{equation} \label{eq:beta}
 \bhbeta := \bhSigma_{xx}^{-1}
 \bhSigma_{xy}\;\;\;
 \mbox{ followed by }\;\;\;
 \halpha = \hmu_y - 
 \bhmu_x^{{\btop}}\bhbeta \;.
\end{equation}
The robustness of this estimator stems
from the 2SGS method, which is known to 
perform quite well in many situations. 
On the other hand, this approach 
requires the covariance matrix to be 
invertible so it is not applicable to 
higher-dimensional settings that the 
shooting S version of
\cite{bottmer2022sparse} can handle.

\begin{table}[ht]
\centering
\vspace{2mm}
\begin{tabular}{rccccc}
\hline
p & Classical & LTS & S & Shooting S \\ 
\hline
5 & 0.0008 & 0.0442 & 0.0440 & 0.5525 \\ 
10 & 0.0011 & 0.0720 & 0.0947 & 1.1577 \\ 
20 & 0.0018 & 0.1558 & 0.3650 & 3.2820 \\ 
50 & 0.0034 & 0.8009 & 1.7363 & 12.417 \\ 
\hline
\end{tabular}
\caption{Linear regression: 
   computation times in seconds 
   for a dataset with $n = 1000$ and 
   $p$ regressors.}
\label{tab:reg}
\end{table}

Table~\ref{tab:reg} shows the
computation times of several regression
methods, using the same setup as in
Tables~\ref{tab:flagging} 
and~\ref{tab:cov}.
The classical least squares regression
is compared to the casewise robust
Least Trimmed Squares regression 
\citep{rousseeuw1984least} and S 
regression \citep{RY1984}, and 
to the cellwise robust shooting S 
method of \cite{ollerer2016shooting}. 
The latter is the slowest but still
quite feasible, even though its
computation is in pure \texttt{R}.
The proposal of \cite{leung2016robust}
is not shown because its computation 
time is that of the 2SGS method in
Table~\ref{tab:cov}, the additional
computation being very quick.

An interesting approach was introduced 
by \cite{agostinelli2016composite} for 
the more general problem of robust 
estimators for linear mixed models. 
They propose a 
composite estimator which minimizes 
the scales of Mahalanobis distances 
on pairs of variables. 
The use of pairwise Mahalanobis 
distances makes this approach more 
robust against cellwise outliers than 
casewise S-estimators. 
It would be interesting to evaluate the 
robustness of this approach against 
adversarial cellwise outliers.

The proposal of 
\cite{filzmoser2020cellwise} starts
from a casewise robust method such 
as MM-regression. Then \textbf{SPADIMO} 
\citep{Debruyne2019} is 
applied to the flagged cases in the
hope of identifying the cells 
responsible for their outlyingness.
Next, these flagged cells are imputed 
by nearest neighbors applied to the 
unflagged cells, and the imputed data 
is used to update the residuals of the 
regression. 
Finally, casewise weights are obtained 
from the standardized residuals, and
used in a weighted least squares 
regression to obtain updated regression 
coefficients. 
This three-step procedure is iterated.  
As the algorithm relies on a casewise 
robust initial estimator, it cannot be 
cellwise robust in a general sense. 
However, it is possible that cellwise 
robustness could be obtained by 
replacing the initial estimator.

Finally, we mention some recent related 
proposals. 
\cite{su2021robust} first compute a 
robust covariance matrix of the joint
distribution of the pairs $(\bx_i, y_i)$.
For this they use a method based on
pairwise correlations as mentioned above,
using the approach of  
\cite{gnanadesikan1972robust} or the  
Gaussian rank correlation studied by 
\cite{boudt2012gaussian}.
They then plug the square root of this 
covariance matrix into the adaptive 
lasso objective to obtain variable
selection.
\cite{toka2021robust} propose a 
three-step procedure consisting of (i) 
identifying marginal cellwise outliers 
using robust z-scores and setting them
to missing, (ii) running the casewise 
robust MCD estimates of location and
covariance on the complete cases, and 
then imputing the NA's as in the E-step 
of the EM algorithm, and (iii) applying 
a robust lasso regression to the 
imputed data. 
\cite{saraceno2021robust} carry out
robust seemingly unrelated regressions 
using several univariate MM-regressions, 
and apply the 2SGS method of 
\cite{Agostinelli2015} to the residual
matrix.
For regression with compositional 
covariates, \cite{vstefelova2021robust} 
propose a four-step procedure: 
(i) detecting outlying cells by \textbf{DDC}, 
(ii) imputing the outlying cells that 
are not in outlying rows, 
(iii) running rowwise robust regression
on the imputed data, and 
(iv) performing inference based on multiple 
imputation.

We end the section with a brief discussion.
While there have been several approaches 
to cellwise robust regression, there is 
still much room for further research. 
First, not much attention has been devoted
to outlier detection so far. 
For casewise robustness, there is a taxonomy 
of outlier types based on whether $\bx_i$
or the residual of $y_i$ is outlying, or 
both, see \cite{rousseeuwZomeren1990}. 
The detection of such casewise outliers 
depends crucially on having a reliable 
robust fit $(\halpha,\bhbeta)$. 
But in many of the existing proposals, 
cellwise outliers are detected in a kind of 
pre-processing step, and the robust fit 
is not yet used to detect them. 
Second, with few exceptions, most of the
proposals do not yet have provable 
robustness guarantees, such as a high 
breakdown value under cellwise 
contamination. 
This may be due to the lack of a natural 
framework such as those for casewise 
robust regression, or it may be due to 
the general complexity of dealing with
cellwise outliers.

\subsection{Principal component analysis}

Principal component analysis (PCA) is an 
essential tool of multivariate statistics. 
Several formulations of classical PCA
exist, e.g. based on the spectral 
decomposition of the covariance matrix
or the singular value decomposition of
the centered data matrix.
One way of looking at it is that PCA
approximates the dataset $\bX$ 
with $d$ columns by a new dataset 
$\bhX$ that lies in a $q$-dimensional
affine subspace of dimension $q < d$ 
such that the Frobenius norm 
$||\bX - \bhX||_F$ is minimized. 
Several proposals exist for casewise 
robust PCA, including spherical
PCA \citep{locantore1999robust}, PCA
based on least trimmed squares orthogonal
regression \citep{maronna2005principal}, 
the PCAgrid projection pursuit method 
\citep{croux2007}, and the hybrid 
ROBPCA method \citep{hubert2005robpca}. 
There are also approaches for casewise 
robust sparse PCA, see 
\cite{croux2013sparsePCA}, 
\cite{HubertROSPCA},
and \cite{wang2020sparse}.

For cellwise outliers, things are more 
difficult. One complication stems from 
the fact that classical PCA projects the 
data orthogonally on a lower dimensional
subspace. Projecting a case (row) with 
cellwise outliers is problematic, as the 
outlying cells can propagate over all
cells. Therefore, whenever the algorithm
makes a projection, the outlying cells
should be addressed.

An early proposal for the related 
problem of cellwise robust factor
analysis was made by \cite{croux2003RAR}.
The first cellwise robust PCA method
was proposed by \cite{maronna2008robust}. 
Instead of minimizing the Frobenius norm 
of $||\bX - \bhX||_F$ they minimize a 
robust approximation error by applying 
a bounded function $\rho$ to each cell 
of $\bX - \bhX$ and summing the 
resulting errors. The minimization 
is carried out by iteratively solving 
weighted least squares problems.
The bounded $\rho$ function guarantees
that a single cell 
cannot dominate the approximation error.
A different proposal for cellwise robust 
PCA was made by \cite{she2016robust}.
The MacroPCA method proposed by
\cite{hubert2019macropca} aims to be an
all-in-one method capable of handling
missing values as well as being robust
to cellwise and rowwise outliers. 
It starts from \textbf{DDC} and 
incorporates elements of ROBPCA. 
It imputes
cells that caused data points to lie 
far from the fitted subspace, and
produces cellwise residuals that can
be plotted.

\begin{table}[ht]
\centering
\vspace{2mm}
\begin{tabular}{rccccc}
\hline
d & Classical & ROBPCA  & PCAgrid & MacroPCA\\ 
\hline
10 & 0.0043 & 0.1091 & 0.2216 & 0.1551 \\ 
50 & 0.0161 & 0.2414 & 1.8707 & 0.4163 \\ 
100 & 0.0413 & 0.3641 & 4.0087 & 0.9460 \\ 
500 & 1.0227 & 2.9865 & 22.727 & 12.142 \\ 
1000 & 3.3084 & 8.6023 & 56.335 & 13.778 \\ 
\hline
\end{tabular}
\caption{PCA: computation times in 
seconds for a dataset with $n = 1000$ 
and dimension $d$.}
\label{tab:pca}
\end{table}

The computation times of several 
of these PCA methods are shown in 
Table~\ref{tab:pca}. The setup was the
same as in Tables \ref{tab:flagging}
to~\ref{tab:reg}, but going up to
higher dimensions $d$. The task was to 
compute the first 10 components.
The casewise robust ROBPCA and PCAgrid 
methods take more time than classical 
PCA but are still fast enough. The 
MacroPCA method is slower than ROBPCA 
but faster than PCAgrid, illustrating
that cellwise robust PCA is 
computationally affordable.

\subsection{Clustering}

Several methods for clustering data with 
cellwise outliers have been proposed. 
The earliest was \cite{Farcomeni2014}, 
who used a mixture model with an 
additional $n \times d$ parameter matrix
of zeroes and ones, indicating which 
cells should be flagged and set to 
missing. 
The goal was to maximize the observed 
likelihood of the unflagged cells. 
The method was implemented in the 
\textsf{R} package \texttt{snipEM} 
\citep{Farcomeni2019} which is no longer 
maintained. 
\cite{garcia2021cluster} consider a
model in which clusters are linear,
i.e. they lie near lower-dimensional
subspaces. Each cluster is then fitted
by a cellwise robust PCA in the style 
of \cite{maronna2008robust} discussed 
above. They minimize the robust scales
of the coordinatewise residuals as in
least trimmed squares regression. 
The algorithm iterates the following
three steps: (i) updating the subspace 
parameters by robust PCA, (ii) updating 
the weights indicating which cells are
deemed outlying, and (iii) updating 
the group memberships, similar to the
usual k-means clustering algorithm.

\subsection{Time series}\label{sec:TS}

We are not aware of existing work 
on applying cellwise robust methods to
time series, but we will describe a
setting in which they prove useful, 
the fitting of autoregressive models. 
Suppose we have a univariate time 
series $y_t$ for $t=1,\ldots,n$.
The autoregressive model AR($p$) of 
order $p$ then assumes that
\begin{equation*}
  y_t = \beta_1 y_{t-1} +
    \ldots + \beta_p y_{t-p} + 
    \varepsilon_t
\end{equation*}
for $t=p+1,\ldots,n$ where the
$\varepsilon_t$ are i.i.d. according
to a Gaussian distribution with mean
zero and standard deviation $\sigma$.
For simplicity we assume that the 
model has no intercept term.
Note that there are $n-p$ complete
sets $(y_t, y_{t-1},\ldots,y_{t-p})$
which can be combined as successive 
rows in an 
$(n-p) \times (p+1)$ matrix $\bZ$.
The AR($p$) model can thus be seen as 
a linear model with the first column
of $\bZ$ as response and the other
columns as regressors, and various
classical methods exist to estimate
the parameter vector $\bbeta =
(\beta_1,\ldots,\beta_p)$ from
$\bZ$. The simplest of these is just
the LS regression estimator,
sometimes called `naive OLS' in this 
context.

Now consider an outlying 
value $y_t$\,.
In the AR($p$) model it will occur
$p+1$ times: once as response, 
and once in each of the $p$ lags.
In other words, one outlying $y_t$
affects $p+1$ rows of $\bZ$, each time
in a single cell.
Since the breakdown value of any
affine equivariant regression method
is at most $0.5$, using such a method 
in this setting has a breakdown value 
of at most $1/(2p+2)$ which can be 
unpleasantly low. (And if the original
series is first differenced, e.g. to
remove a seasonal effect, this upper
bound is cut in half again.)
The odds against robust estimation of 
$\bbeta$ are thus stacked.
For a discussion of this issue and a
review of the literature on types of
outliers in AR($p$) models see 
Section 7.2 of \cite{RL1987}.

This is a setting where 
cellwise robust methods can shine.
As an illustration, we generated an
AR(3) time series with parameters
$\bbeta = (0.5,\,0.2,\,0.2)$,
$\sigma = 1$ and length $n=1000$. 
We then created a kind of
`day of the week effect' by replacing
every 7th entry $y_t$ by the outlying
value 10.
(The \textsf{R} script generating the 
data and the subsequent analysis is
available from the webpage
\url{https://wis.kuleuven.be/statdatascience/robust/papers-2020-2029}).
The original time series
thus has 14\% of outliers, but they 
contaminate 569 rows of $\bZ$ out of 
997, so 57\% of them.
The classical estimate of $\bbeta$
becomes $(0.11,0.11,0.08)$ which is
far from $(0.5,\,0.2,\,0.2)$, and 
the estimate of the error standard
deviation $\sigma=1$ is 3.79\,.
Casewise robust regression does not
work here either, with LTS yielding
$\bhbeta = (0.73, 0.02, 0.01)$ and
$\hsigma = 1.85$, and the default
MM regression in \textsf{R} giving
$\bhbeta = (0.14, 0.14, 0.12)$ and
$\hsigma = 1.72$\,.
Fortunately, cellwise robust methods
do give good results.
The regression method 
of \cite{leung2016robust}
first computes a robust covariance
matrix $\bhSigma$ of $\bZ$ by 2SGS
and then applies~\eqref{eq:beta}, 
yielding
$\bhbeta = (0.52, 0.21, 0.17)$.
Moreover, from 
$\hsigma^2 = \bhSigma_{yy} - 
\bhbeta^{\btop} \bhSigma_{xx}
\bhbeta$ one obtains $\hsigma = 0.96$.
Applying the same formulas to the
cellMCD covariance matrix gives
similar results, with 
$\bhbeta = (0.53, 0.18, 0.19)$ and
$\hsigma = 0.92$\,.

To conclude, the autoregressive 
model can by its nature create many
outlying cells, which handicap
casewise robust methods whereas
cellwise robust methods can deal
with them.
 
\subsection{Various other settings}

\cite{aerts2017cellwise} construct a
cellwise robust version of linear 
and quadratic discriminant analysis.
To this end they estimate a cellwise
robust precision matrix for each class
as in \cite{croux2016robust}, and
plug these into the assignment rule.
Further work could deal with the 
situation that new data points to be
classified may also contain
outlying cells.

In functional data analysis, a case
is not a row of a matrix but a function 
such as a curve. 
In that setting, cellwise outliers
correspond to local deviations (like
spikes) in a curve, rather than global
outlyingness of the entire curve.
In the taxonomy of \cite{Hubert2015SMA}
of functional outliers these are
called {\it isolated outliers}, and 
methods have been constructed to detect 
them.
\cite{garcia2021cluster} apply
cellwise robust methods to functional 
data.

\section{Breakdown values}
\label{sec:bdv}

The breakdown value plays an important
role in the study of casewise robust 
methods, because a positive breakdown 
value is a necessary condition for 
robustness. 
In several models, natural equivariance 
properties imply an upper bound of 50\% 
on the breakdown value, and methods 
have been developed that get close to 
this upper bound.

In the setting of cellwise robustness,
\cite{alqallaf2009} define the 
asymptotic cellwise breakdown value of 
a location estimator.
Here we will focus on finite-sample 
breakdown values in the sense of 
\cite{donoho1983}
and \cite{lopuhaa1991breakdown}.
Below we will study two statistical 
models of practical importance, and 
show that under some conditions that 
appear natural at first sight it is 
impossible to obtain a high cellwise
breakdown value.

\subsection{Breakdown of location and
            covariance estimators}
Consider a finite sample consisting of 
$n$ data points in $d$ dimensions,
denoted by its $n \times d$ data matrix
$\bX$.
We want to estimate its unknown location
vector by the estimator $\bhmu$.
Then we define the finite-sample 
cellwise breakdown value of $\bhmu$ at 
$\bX$ as the smallest fraction of
cells per column that need to be replaced 
to carry the estimate outside all bounds.
Formally, denote by $\bX^m$ any corrupted 
sample obtained by replacing at most $m$ 
cells in each column of $\bX$ by arbitrary 
values. 
Then the {\it finite-sample cellwise 
breakdown value} of the location estimator 
$\bhmu$ at $\bX$ is given by 
\begin{equation} \label{eq:bdvloc}
  \varepsilon^*_n(\bhmu, \bX)=
  \min \left\{\frac{m}{n}:\;
	\sup_{\bX^m}{\left|\left|\bhmu(\bX^m) - 
	\bhmu(\bX)\right|\right|} = 
	\infty\right\}.
\end{equation}
This is a natural finite-sample version
of the asymptotic cellwise breakdown 
value defined by \cite{alqallaf2009}.

Let us start with a trivial remark.
If an estimator has an upper bound 
on its casewise breakdown value,
then that upper bound also holds for 
its cellwise breakdown value. Indeed, 
a contaminated dataset obtained by 
replacing $m$ cases by other cases
can also be created by replacing the
cells of the replaced cases by those
of the replacing cases, which affects
$m$ cells in each column.
Therefore \eqref{eq:bdvloc} is at
most $\lfloor (n+1)/2 \rfloor/n$ for
any translation equivariant location
estimator due to Theorem 2.1 of
\cite{lopuhaa1991breakdown}.

We will see that some conditions that
appear natural imply a rather low
cellwise breakdown value.
One such condition is
\begin{equation}
\begin{aligned}\label{eq:locEFP}
 \mbox{if all points of } \bX
 \mbox{ lie in a lower-dimensional}\\ 
 \mbox{ affine subspace, then } 
 \bhmu(\bX)
 \mbox{ lies in that subspace as well.}
\end{aligned}
\end{equation}
This is a very weak exact fit property,
which nevertheless restricts the 
breakdown value as shown by the 
following result.

\begin{proposition} \label{prop:loc}
The cellwise breakdown value of any 
location estimator $\bhmu$
satisfying~\eqref{eq:locEFP} is 
bounded by
\begin{equation} \label{eq:locbdv}
  \varepsilon^*_n(\bhmu, \bX) =
	\left\lceil \frac{n}{d} 
	\right\rceil/n\, \approx 
	\frac{1}{d}
\end{equation}
at any dataset $\bX$.
\end{proposition}

The proof uses the fact that
the data can be moved to a hyperplane
by replacing no more than
$\lceil n/d \rceil$ cells in each 
variable. All proofs of this section 
can be found in the Appendix.

\cite{alqallaf2009} prove an upper 
bound on the asymptotic cellwise 
breakdown value of a location 
estimator that also goes to zero
like $1/d$.
For this they require two conditions.
One of these is the so-called 
$\delta$-consistency of the estimator.
The other is affine equivariance of
$\bhmu$, which says that for all 
$d \times 1$ vectors $\bv$ and all 
$d \times d$ nonsingular matrices $\bA$ 
it holds that 
\begin{equation} \label{eq:affloc}
 \bhmu(\bv \bone_d^{\btop}+\bX\bA)
  = \bv + \bA^{\btop} \bhmu(\bX)  
\end{equation}
where $\bone_d$ is the $d \times 1$
vector of ones.
In words, the estimator is equivariant
for translations and linear 
transformations.
The transposed $\bA^{\btop}$ 
in~\eqref{eq:affloc} is due
to the fact that $\bX$ contains the
data points as rows, not columns.

Our result is more general, in that
we do not need $\delta$-consistency,
and instead of affine equivariance
we can use orthogonal equivariance.
This is a weaker form of equivariance
which only requires~\eqref{eq:affloc}
for orthogonal matrices $\bA$.

\begin{corollary} \label{prop:loc2}
Any orthogonally equivariant location 
estimator satisfies 
property~\eqref{eq:locEFP}, 
and hence the upper bound of 
Proposition~\ref{prop:loc} applies.
\end{corollary}

All this illustrates that orthogonal 
equivariance, and even the weaker 
condition~\eqref{eq:locEFP}, are  
unpromising properties for location 
estimators that want to be cellwise robust.
But for location, one can easily avoid 
this problem by combining coordinatewise
robust location estimates.

\begin{figure}[!ht]
\centering
\includegraphics[width=0.7\columnwidth]
  {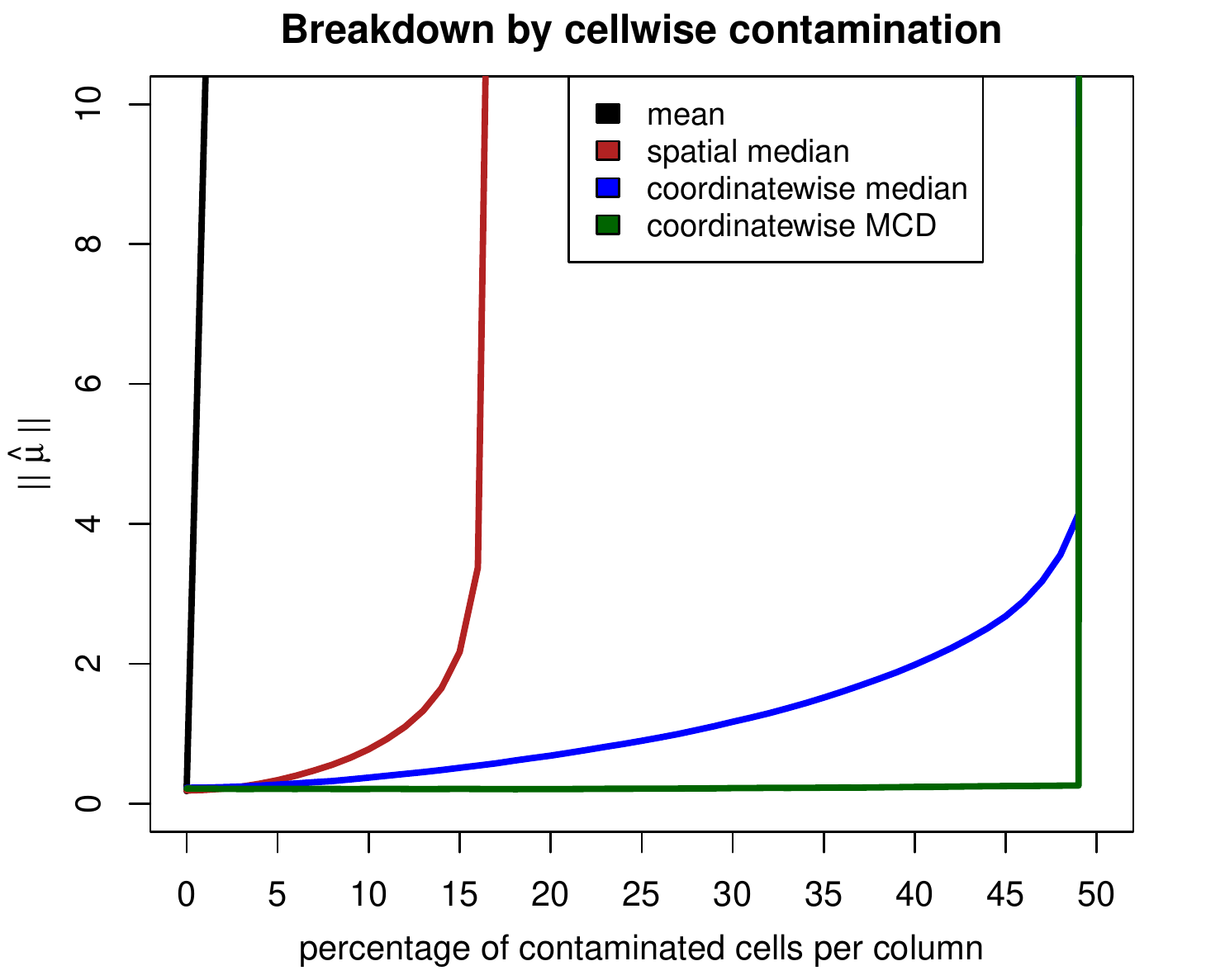}
\caption{Empirical illustration of breakdown
   due to cellwise contamination, for
   $n=100$ cases in $d=4$ dimensions.}
\label{fig:breakdownplot}
\end{figure}

We provide a simple empirical 
illustration. The experiment went as follows.
First $n=100$ points were generated from
the standard Gaussian distribution in
$d=4$ dimensions.
These have $k=0$ outlying cells per
column. Next, the number of outlying
cells per column was incremented in steps 
of 1 such that the cases with outlying
cells overlap as little as possible,
so when $k=25$ all cases have one 
outlying cell. The outlying cells were 
all given the value 500. 
On each contaminated dataset the
following location estimators were run:
the arithmetic mean, the spatial median
which minimizes 
$\sum_{i=1}^n ||\bx_i - \bhmu||$,
the coordinatewise 
median~\eqref{eq:comed},
and the coordinatewise univariate
MCD estimator.
Figure~\ref{fig:breakdownplot} plots
the euclidean norm $||\bhmu||$ of each
estimator as a function of the percentage
of contaminated cells per column.
It shows the averaged norms over 
200 replications. The horizontal
axis ends at 50\%, corresponding to the
upper bound on the breakdown value of
translation equivariant estimators.

In Figure~\ref{fig:breakdownplot} 
we see that the mean is the first to move
far away. In fact, if we would let $n$
go up and the value of every outlying cell 
go from 500 to infinity, the line of the 
mean would become vertical, since 1 
cellwise outlier per column is enough for
breakdown.
In contrast, the norm of the spatial
median stays low for a long time before
going up. (In fact, for exactly the
fraction $1/d = 25\%$ of outlying cells 
per column the spatial median attains 
norm 250, because then all datapoints 
lie near the hyperplane 
$x_1 + x_2 + x_3 + x_4 = 500$.)
Note that the spatial median is not
affine equivariant but it is 
orthogonally equivariant, so
Corollary~\ref{prop:loc2} applies.

The coordinatewise median escapes 
the upper bound $1/d$ because it does 
not satisfy the 
condition~\eqref{eq:locEFP}. Indeed,
\cite{RL1987} already gave a simple
counterexample (on page 250) where the
coordinatewise median lies outside 
the hyperplane whenever $d>2$.
(For $d \leqslant 2$ the bound is
trivially satisfied.)
Since the median of each variable
only breaks down when the variable has
50\% or more outlying values, the
same holds when these separate medians 
are combined in~\eqref{eq:comed}.
The reasoning is the same for 
combining the univariate MCD estimates
of each coordinate.

Cellwise robust covariance estimation
is harder than location.
Estimators $\bhSigma$ of a covariance
(scatter) matrix can break down in two
ways.
Analogously to the casewise setting, we 
can define the {\it cellwise explosion 
breakdown value} of the covariance 
estimator $\bhSigma$ as
\begin{equation} \label{eq:explosion}
  \varepsilon^+_n(\bhSigma, \bX)=
  \min \left\{\frac{m}{n}:\;
	\sup_{\bX^m}\lambda_1(\bhSigma) = 
	\infty\right\}
\end{equation}
where $\lambda_1$ denotes the largest 
eigenvalue.
Moreover, we define the {\it cellwise 
implosion breakdown value} of 
$\bhSigma$ as
\begin{equation} \label{eq:implosion}
  \varepsilon^-_n(\bhSigma, \bX)=
	\min \left\{\frac{m}{n}:\;
	\inf_{\bX^m}\lambda_d(\bhSigma) = 0
	\right\}
\end{equation}
where $\lambda_d$ is the smallest 
eigenvalue. 
Implosion is just as bad as explosion,
since a singular covariance matrix
is not invertible so one cannot 
compute Mahalanobis distances, carry 
out discriminant analysis, and so on.

Also in this setting, a good cellwise 
breakdown value is harder to obtain than
its casewise counterpart. For instance,
the casewise implosion breakdown value of
the classical covariance matrix 
at a typical dataset is very 
high, in fact it is $(n-d)/n \approx 1$.
This is because whenever $d+1$ of the
original data points are kept, the 
classical covariance remains nonsingular.
In stark contrast, its
{\it cellwise} implosion breakdown
value is quite low.
This even holds for all covariance
estimators that satisfy the following
intuitive condition analogous 
to~\eqref{eq:locEFP}:
\begin{equation}\label{eq:covEFP}
\begin{aligned}
 \mbox{if all points of } \bX
 \mbox{ lie in a lower-dimensional}\\ 
 \mbox{affine subspace, then } 
 \bhSigma(\bX)
 \mbox{ is singular.}
\end{aligned}
\end{equation}

\begin{proposition} \label{prop:impl}
The cellwise implosion breakdown value
of any covariance estimator $\bhSigma$
\mbox{satisfying} \eqref{eq:covEFP} is 
bounded by
\begin{equation} \label{eq:impl}
  \varepsilon^-_n(\bhSigma, \bX) =
	\left\lceil \frac{n-1}{d} 
	\right\rceil/n\, \approx 
	\frac{1}{d}
\end{equation}
for any dataset $\bX$.
\end{proposition}

This result appears to be new, and it is
also implied by equivariance.
Recall that a covariance estimator 
$\bhSigma$ is affine equivariant if for
all $d \times 1$ vectors $\bv$ and all 
$d \times d$ nonsingular matrices $\bA$ 
it holds that 
\begin{equation} \label{eq:aff}
 \bhSigma(\bv \bone_d^{\btop}+\bX\bA)
  = \bA^{\btop} \bhSigma(\bX) \bA\;. 
\end{equation}

\begin{corollary} \label{prop:impl2}
Any affine equivariant covariance 
estimator $\bhSigma$ satisfies
property~\eqref{eq:covEFP}, and 
hence the upper bound of 
Proposition~\ref{prop:impl} applies.
\end{corollary}

The upper bound 
$\lceil (n-1)/d \rceil/n\approx 1/d$ 
in Corollary~\ref{prop:impl2} is a lot
lower than that inherited from the 
casewise breakdown value of affine 
equivariant
covariance estimators, which is 
$\lfloor (n-d+1)/2\rfloor /n 
\approx 0.5$\,.
The fact that implosion breakdown can
happen easily in the cellwise setting
was not mentioned in the literature 
before. 
For the cellwise robust covariance 
estimators of \cite{ollerer2015robust}, 
only their explosion breakdown value
was stated.
The above results make
it useful to rule out 
implosion by a prior constraint like 
$\lambda_d(\bhSigma) \geqslant a$, or
from a formulation in which 
$\bhSigma$ is a convex 
combination of two matrices, one of 
which is the identity matrix with a
small coefficient.
Both approaches would be at odds with 
affine equivariance, but here we are in 
the cellwise setting.
This mechanism has also been employed
in casewise settings without affine 
equivariance, see e.g. \cite{ollila2014}.
	
Note that in spite of the above negative
results it is possible to construct 
estimators of covariance that achieve a 
higher cellwise breakdown value, as in
\cite{cellMCD}. Their proposal for 
a cellwise robust MCD estimator attains 
a breakdown value of $(n-h+1)/n$ where 
$h$ can be as low as 
$\lfloor n/2\rfloor + 1$ but is
recommended to be chosen as $h = 0.75n$.
It can thus reach an asymptotic 
breakdown value of $0.5$, which 
is the highest possible.

\subsection{Breakdown in regression}	

For linear regression we use the 
model~\eqref{eq:reg} but write it a
little differently as
\begin{equation} \label{eq:bigreg}
 \by = \bX \bgamma + \mbox{noise} 
\end{equation}
where $\by = (y_1,\ldots,y_n)$ is the 
$n \times 1$ column vector of responses,
and the rows of $\bX$ are of the
form $(1,\bx_i^{\btop})$. 
The $(p+1) \times 1$ vector 
$\bgamma = (\alpha,\beta_1,\ldots, 
\beta_p)^{\btop}$ combines
the intercept and the slopes. 
We also combine the actual data 
(that is, everything but the first
column of $\bX$ which is constant)
in the $n \times (p+1)$ matrix $\bZ$
whose $i$-th row is
$(x_{i1},\ldots,x_{ip},y_i)$.

We now contaminate the dataset
$\bZ$. By $\bZ^m$ we
mean any corrupted set of covariates 
obtained by replacing at most $m$
cells in each column of $\bZ$ 
by arbitrary values.
We define the finite-sample cellwise 
breakdown value of the regression 
estimator $\bhgamma$ at the dataset
$\bZ$ as 
\begin{equation} \label{eq:bdvreg}
	\varepsilon^*_n(\bhgamma, \bZ)=
  \min \left\{\frac{m}{n}:\;
	\sup_{\bX^m}{\left|\left
	|\bhgamma(\bZ^m) -
	\bhgamma(\bZ)\right|\right|} =
	\infty\right\}.
\end{equation}
	
Many existing regression methods
$\bhgamma$ satisfy the following 
intuitive property:
\begin{equation}\label{eq:regEFP}
\begin{aligned}
 \mbox{if } y_i = 
 \bx_i^{\btop} \bgamma_0 
 \mbox{ for all $i=1,\ldots,n$ and } 
 \bX 
 \mbox{ is not}\\
 \mbox{in a lower-dimensional}
 \mbox{ subspace, then }
 \bhgamma = \bgamma_0\;.
\end{aligned}
\end{equation}
The condition that $\bX$ does not lie
in a lower-dimensional vector subspace
serves to avoid multicollinearity,
which causes problems for many 
regression methods.

\begin{proposition} \label{prop:reg}
The cellwise breakdown value of any 
regression estimator $\bhgamma$
satisfying~\eqref{eq:regEFP}
is bounded by
\begin{equation} \label{eq:bdreg}
  \varepsilon^*_n(\bhgamma, \bX, \by)
	\leqslant \left\lceil \frac{n-1}{p+1} 
	\right\rceil/n\, \approx 
	\frac{1}{p+1}
\end{equation}
at any dataset $\bZ$.
\end{proposition}
	
Also the weak condition~\eqref{eq:regEFP}
is satisfied under natural equivariance 
properties. The first of these is 
{\it regression equivariance}, which 
requires that for all $(p+1)\times 1$ 
vectors	$\bv$ it holds that
\begin{equation}\label{eq:regeq}
  \bhgamma(\bX,\by+\bX\bv)=
	\bhgamma(\bX,\by)+\bv
\end{equation}
for all datasets $(\bX,\by)$. This
intuitive condition is similar to 
translation equivariance for 
location estimators.
The second property is {\it scale
equivariance}, meaning that for any
constant $c$ we have
\begin{equation}\label{eq:scaleq}
  \bhgamma(\bX,c\,\by)=
			  c\,\bhgamma(\bX,\by)\;.
\end{equation}
We do not even need
affine equivariance, which says that
$\bhgamma(\bX \bA,\by)=
\bA^{-1}\bhgamma(\bX,\by)$
for all nonsingular matrices $\bA$.
	
\begin{corollary} \label{prop:reg2}
Any estimator $\bhgamma$ that is
regression equivariant~\eqref{eq:regeq} 
and scale equivariant~\eqref{eq:scaleq} 
satisfies property~\eqref{eq:regEFP}, 
and hence the upper bound of 
Proposition~\ref{prop:reg} applies.
\end{corollary}

We could easily write down analogous
results for implosion of the scale
estimator of the noise term.

In conclusion, if we want a regression
estimator with a better cellwise 
breakdown value we have to let go of 
the natural combination of equivariance 
properties~\eqref{eq:regeq}
and~\eqref{eq:scaleq}, and even of the 
intuitive condition~\eqref{eq:regEFP}.
When working with high-dimensional data
we are used to penalizations that cause
the regression to deselect some of the 
covariates, but even for 
lower-dimensional data one may have to 
resort to such techniques, 
or novel ones.
	
\section{Cellwise robust correspondence 
         analysis} \label{sec:CA}
				
In this section we present the first
cellwise robust method for correspondence 
analysis.

\subsection{Classical correspondence 
            analysis}
Correspondence analysis (CA) was proposed 
by \cite{hirschfeld1935connection} and 
further developed by 
\cite{benzecri1973analyse} as a method 
for analyzing categorical data that can 
be represented by a contingency table. 
The main aim is to construct a so-called 
biplot, a bivariate plot summarizing 
part of the structure in the contingency 
table. This biplot can be thought of as 
analogous to plotting the first two 
principal components of continuous 
numerical data.

Let $\bX$ be an $n\times d$ contingency 
table of counts $x_{ij}$ with 
$i = 1,\ldots,n$ and $j = 1,\ldots,d$.
Denote by $N = \sum_{i=1}^{n}
\sum_{j=1}^{d}\bX_{ij}$ the sum of all 
counts in $\bX$, and let $\bP=\bX/N$ 
be the matrix of relative frequencies. 
Denote by $\br$ the $n \times 1$ column 
vector containing the sums of the rows 
of $\bP$\,, and by $\bc$ the 
$d \times 1$ vector containing the sums 
of its columns. 
Also write $\bD_{\br} = \diag(\br)$ and 
$\bD_{\bc} = \diag(\bc)$. 
Finally, let $\bR = \bD_{\br}^{-1} \bP$ 
be the matrix of standardized row 
profiles, for which each row sums to 1. 
Correspondence analysis is based on
the matrix of weighted centered row 
profiles
\begin{equation}\label{eq:matrixS}
\bS := \bD_{\br}^{1/2}\left(\bR - \bone_d 
 \bc^{\btop}\right)\bD_{\bc}^{-1/2}
\end{equation}
where $\bone_d$ is the $d \times 1$
column vector of ones.
Note that each row of $\bone_d \bc^{\btop}$
also sums to one, so the row totals of
$\bR - \bone_d \bc^{\btop}$ are all zero.
This implies that the columns of 
$\bR - \bone_d \bc^{\btop}$ are 
linearly dependent, and therefore
$\bS$ is not of full rank either.

From here on, $\bS$ is treated as if 
it were a dataset with continuous 
variables (columns). 
Since $\bS$ is singular by
construction, it is useful to carry 
out its singular value decomposition,
yielding
\begin{equation}\label{eq:SVD}
   \bS = \bU \bGamma \bV^{\btop} \,.
\end{equation}
Here $\bU$ contains the left singular 
vectors in its orthogonal columns, 
$\bGamma$ is a diagonal matrix with 
the singular values on the diagonal,
and $\bV$ contains the right singular 
vectors in its orthogonal columns.
If we denote the rank of $\bS$ by
$k$, the matrix $\bU$ is $n \times k$,
$\bGamma$ is $k \times k$, and the
loadings matrix $\bV$ is $d \times k$\,.
Using this singular value decomposition 
we can now represent the rows as well 
as the columns of the data set $\bS$\,.
The principal coordinates of the rows 
are given by 
$\bD_{\br}^{-1/2} \bU \bGamma $, and
those of the columns are given by 
$\bD_{\bc}^{-1/2} \bV \bGamma$.
In a biplot we only show the
first two coordinates of both. 

\subsection{Cellwise robust CA}

In order to construct a cellwise robust
correspondence analysis method, we need 
a robust version of the singular value
decomposition~\eqref{eq:SVD} of $\bS$.
For this purpose we construct a 
modification of the cellwise robust
MacroPCA method of 
\cite{hubert2019macropca}.
Its original version fits the PCA model
\begin{equation} \label{eq:PCA}
  \bS = \bone_n \bmu^{\btop} + 
	\bT \bV^{\btop} + \mbox{noise} 
\end{equation}
where $\bmu$ is the center, the matrix 
of scores $\bT$ is $n \times k$, and 
the loadings matrix $\bV$ is 
$d \times k$. The noise has the same 
dimensions as $\bS$.
In this context the choice of $k$, the
number of principal components, is
important. 
This is because a point is considered
outlying when its orthogonal distance
to the fitted subspace is high, so
if the maximal $k$ were chosen all
points would be in the subspace and
none would be considered outlying.
This also implies that the loadings 
are not nested, that is, the loading 
columns for $k=2$ are not a subset of 
those for $k=3$.
MacroPCA selects $k$ as the
smallest number of components that
explains at least 80\% of the 
variability.

Since our goal is to fit the 
model~\eqref{eq:SVD}, we have added
an option to MacroPCA which fixes the
center $\bmu$ beforehand rather than
estimating it from the data.
For our current purpose we fix $\bmu$
at zero, so the first term 
of~\eqref{eq:PCA} disappears.
Let us denote the estimated loadings
and scores by $\bhV$ and $\bhT$.
For the singular values in the diagonal 
matrix $\bhGamma$ we take the square 
root of $n$ times the eigenvalues 
estimated by MacroPCA. 
Finally, we compute 
$\bhU = \bhT \bhGamma^{-1}$. 
The combination of the resulting 
matrices $\bhU$, $\bhGamma$ and 
$\bhV$ provides a robust version of 
the decomposition~\eqref{eq:SVD}.
Using these matrices, the biplot can 
be drawn as before.

To illustrate this approach to cellwise 
robust correspondence analysis we consider 
two examples studied previously
by \cite{riani2022robust} in the context 
of casewise robust correspondence analysis. 

\subsection{Correspondence analysis of 
            the clothes data}
The first example is a contingency table 
containing trade flows of clothes from
outside the European Union into each 
of its 28 member states (the data is 
from before 2020, so the United Kingdom 
was still included). 
The columns in the contingency table 
in~\cite{riani2022robust} are five 
different price brackets, from lowest 
to highest. 
The contingency table $\bX$ is thus of 
size $28 \times 5$, as are the matrices 
$\bP$ and $\bS$.

\begin{figure}[ht!]
\centering
\includegraphics[height=0.89\textheight]
  {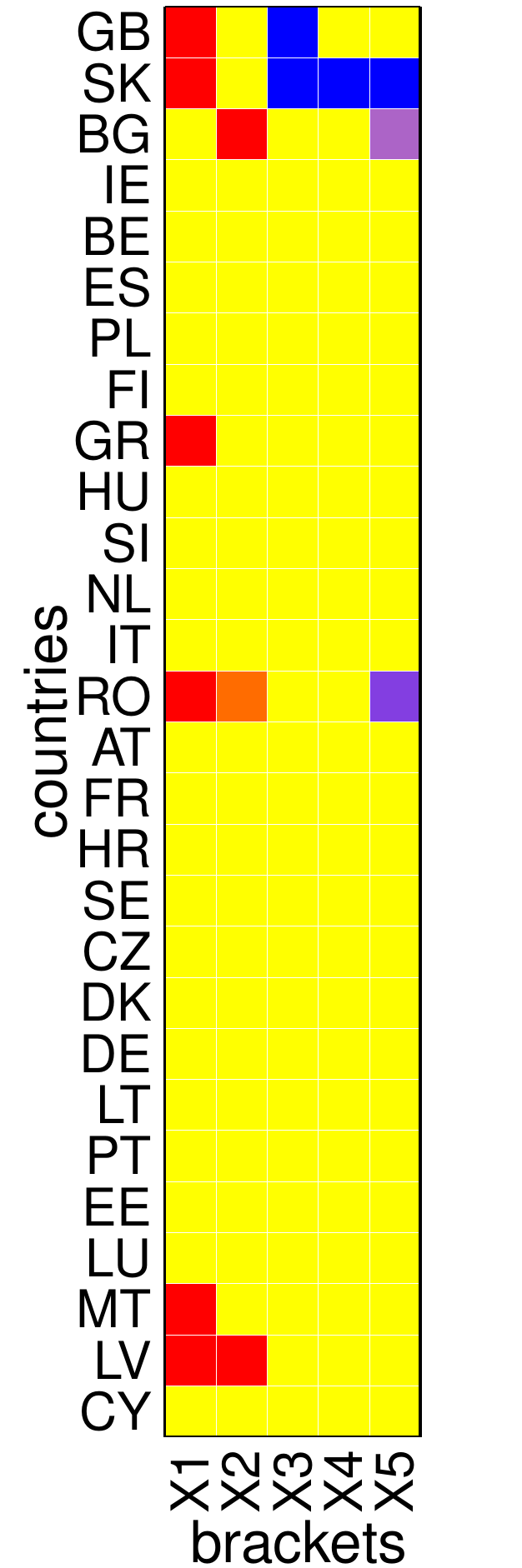}
\includegraphics[height=0.89\textheight]
  {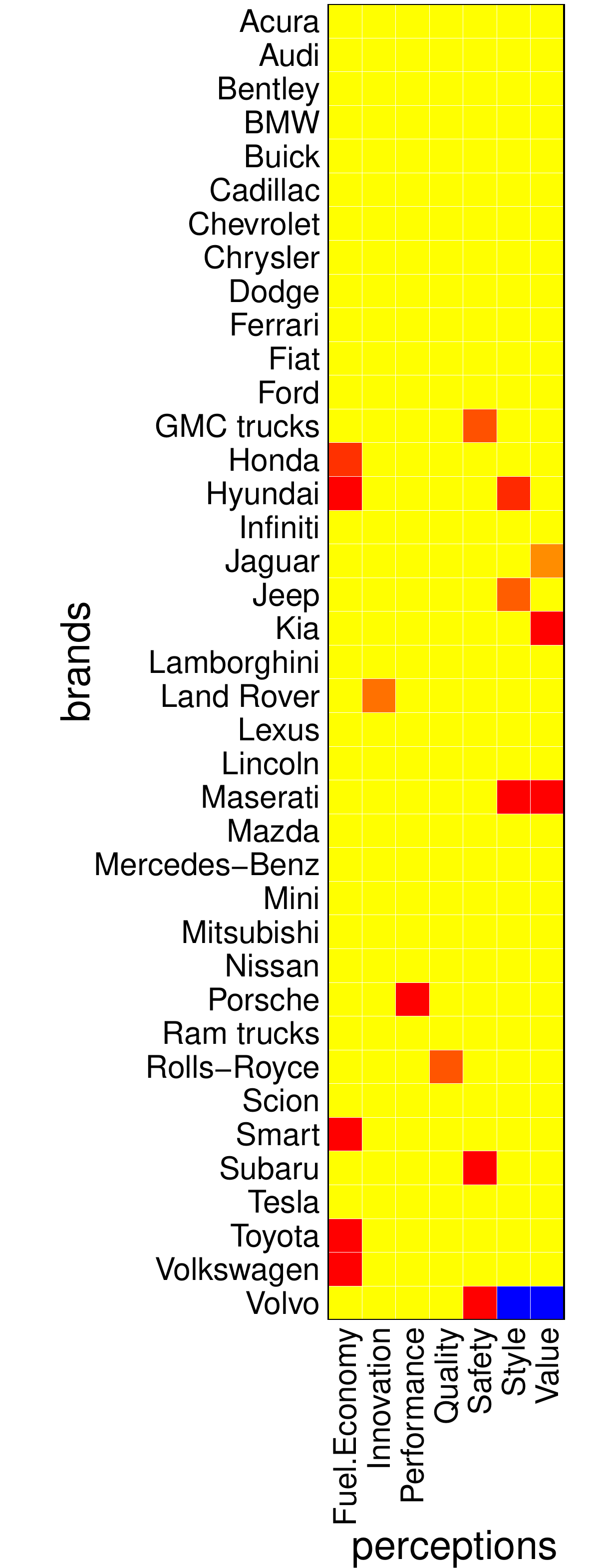}
\caption{Cellmap showing the result of 
DetectDeviatingCells (\textbf{DDC}) on the 
clothes data (left) and on 
the brands data (right).}
\label{fig:ddc}
\end{figure}

As an initial exploration of the data, 
we start by applying \textbf{DDC} to $\bS$. 
(Note that we should not apply \textbf{DDC} 
to the original contingency table $\bX$
because the scales of its rows are not 
comparable, as they are correlated with 
the population size of their countries.)
The \textbf{DDC} algorithm yields a predicted
data matrix $\bhS$, and thus also
cellwise residuals $\bS - \bhS$. 
The left panel of Figure~\ref{fig:ddc}
shows the \textbf{DDC} cellmap of the clothes
data, which graphically represents the 
standardized cellwise residuals.
Red cells indicate standardized 
residuals above 
$\sqrt{\chi^2_{1,0.99}} \approx 2.57$,
with the intensity of the color 
increasing with the actual residual. 
Such cells have a higher value than
predicted.
Blue cells indicate standardized 
residuals below 
$-\sqrt{\chi^2_{1,0.99}}$\,, so their
value was lower than predicted.
The majority of the cells, shown in 
yellow, are roughly in line with 
their predicted values.

The cellmap draws our attention to a
few unusual trade flows.
In particular, the United Kingdom
(with country code GB), Slovakia (SK), 
Greece (GR), Romania (RO), Malta (MT)
and Latvia (LV) import many clothes 
of the lowest price segment,
compared to what would be expected
based on their other trade flows.
These results can be interpreted. 
For instance, it is known that the
United Kingdom imports large 
quantities of very cheap clothing 
from Bangladesh, India and China. 
Also Slovakia imports a lot from 
these countries, and judging from
its blue cells it imports relatively 
few clothes from the upper three 
segments. 

\begin{figure}[ht!]
\centering
\includegraphics[width=0.7\columnwidth]
  {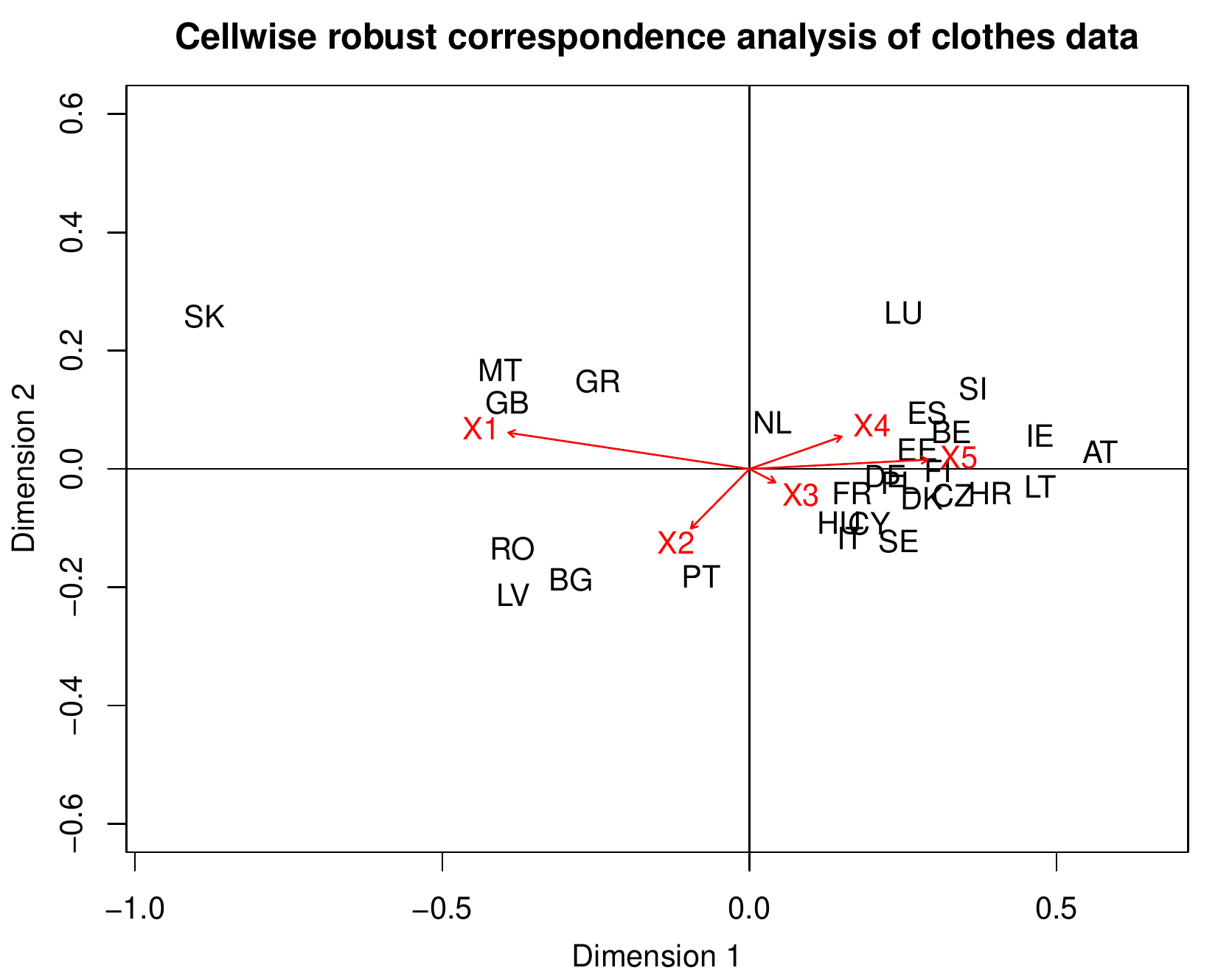}
\caption{Biplot of the clothes data based 
   on the proposed cellwise robust fit.}
\label{fig:clothes_biplot}
\end{figure} 

Next, we apply the zero-center 
version of MacroPCA to $\bS$,
yielding a robust loadings matrix 
$\bhV$, scores $\bhT$, and 
estimated eigenvalues.
From these we derive $\bhGamma$
and $\bhU = \bhT \bhGamma^{-1}$.
Using $\bhU$, $\bhGamma$ and
$\bhV$ yields the biplot in
Figure \ref{fig:clothes_biplot}.
We see that SK is outlying on the left
and loads very strongly onto the category 
of the lowest price bracket $X_1$\,, 
whereas it lies in the opposite direction 
of the variables $X_3$, $X_4$ and $X_5$
corresponding with the three highest 
price brackets. 
We also note the small group of 
countries MT, GB, GR that lie in the
direction of the lowest bracket $X_1$\,,
and the countries RO, LV, BG in the 
direction of the two lowest brackets
$X_1$ and $X_2$\,.
The countries on the right spend more
on the higher brackets.
Overall, this biplot looks quite similar
to the casewise robust biplots in
Figure 4 of~\cite{riani2022robust}. 

\subsection{Correspondence analysis 
            of the brands data}
						
The second example is a contingency table 
summarizing the 2014 Auto Brand Perception 
survey by Consumer Reports (USA), which is 
publicly available on 
\url{https://boraberan.wordpress.com/2016/09/22/}. 
The survey questioned 1578 participants on 
what they considered attributes of 39 
different car brands. 
The list of possible attributes consisted 
of Fuel Economy, Innovation, Performance, 
Quality, Safety, Style, and Value.
Respondents had to select all the 
characteristics from the list which they 
felt applied to a brand. The resulting 
contingency table $\bX$ is thus of size 
$39 \times 7$, so the matrix $\bS$ has 
the same dimensions.

We again start by applying \textbf{DDC} to $\bS$ 
in order to flag outlying cells. 
The resulting cellmap is in the right 
panel of Figure \ref{fig:ddc}, and
reveals some interesting insights. 
We see that Volvo is quite exceptional
in that as many as three of its cells 
are flagged.
Its Safety perception is higher than 
predicted, whereas Style and Value are 
lower than would be expected from its 
other characteristics. 
Hyundai and Maserati have two flagged 
cells. 
Hyundai has rather high Fuel Economy 
and Style given its other characteristics, 
and Maserati seems to have exceptionally 
high Style and Value. 
(One wonders whether sometimes `Value'
was interpreted as `expensive'
rather than `value for money'.) 
Then there are another 12 cases 
with a single cellwise outlier. 
Examples are GMC trucks with high 
perceived Safety, Porsche with high 
Performance, and Toyota and Volkswagen 
with high Fuel Economy. 
In total, 15 out of the 39 brands have
at least one cell flagged by \textbf{DDC}.
The total number of flagged cells is 
19, so only 7\% of all cells, but they
affect close to 40\% of the cases. 
Still, a casewise robust analysis with 
a high-breakdown method can work here 
since over 50\% of the cases have no 
outlying cells.

\begin{figure}[ht!]
\centering
\includegraphics[width=0.7\columnwidth]
  {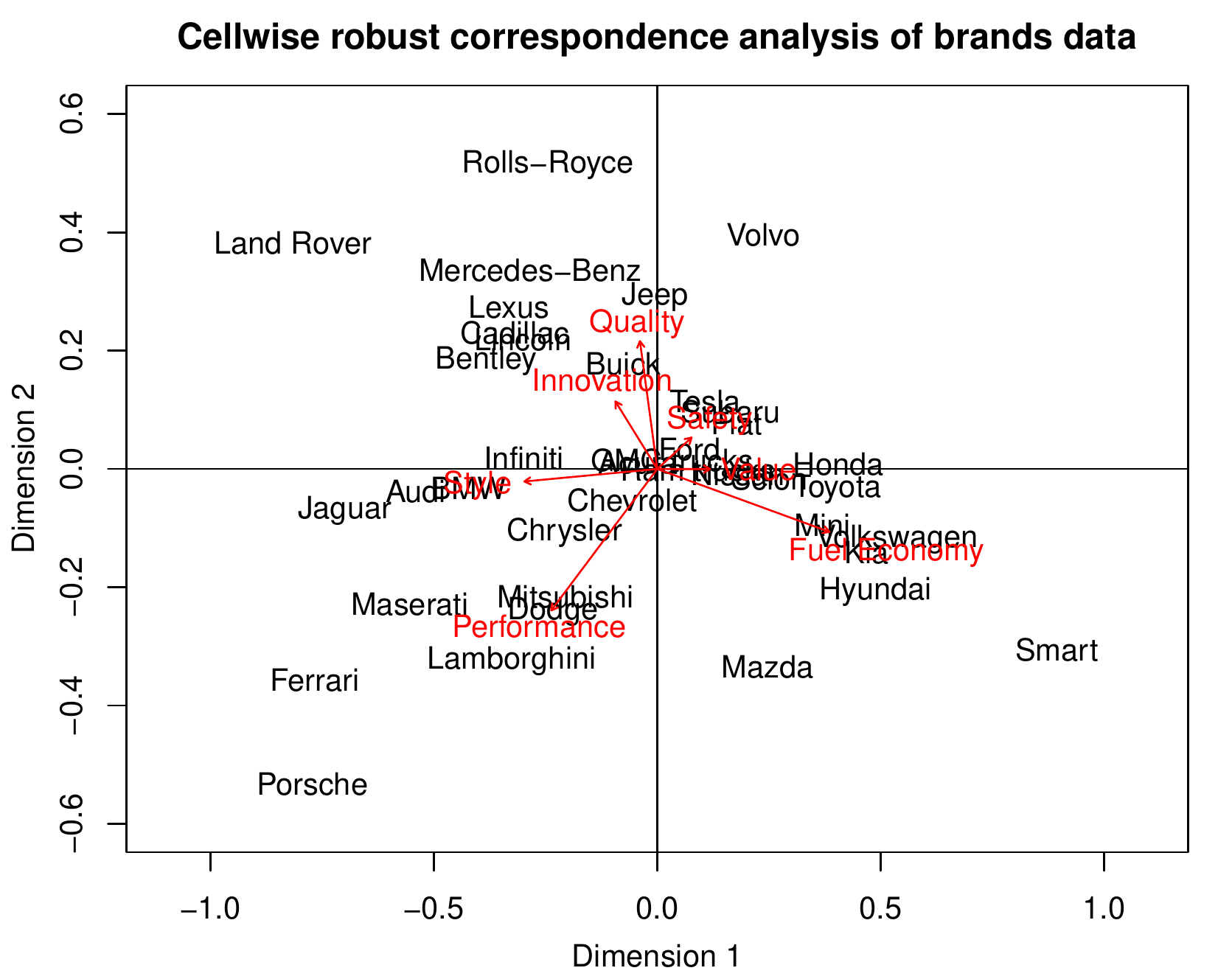}
\caption{Biplot of the brands data based 
   on the proposed cellwise robust fit.}
\label{fig:brands_biplot}
\end{figure} 

Figure~\ref{fig:brands_biplot} shows 
the biplot based on the cellwise robust
MacroPCA method. It is rather similar to
that obtained by~\cite{riani2022robust},
and both robust fits are quite different
from the classical correspondence 
analysis shown in their Figure 6.
We see some intuitively explainable 
structure in 
Figure~\ref{fig:brands_biplot}. 
Volvo is often considered one of the 
safest car brands, and indeed lies in
the direction of the Safety attribute. 
There is a collection of sports cars
(\mbox{Ferrari}, Maserati, Lamborghini, 
Porsche) lying far in the direction of 
the Performance arrow. 
Rolls-Royce and Mercedes-Benz lie in 
the directions of Innovation and Quality, 
and Smart, Hyundai, and Volkswagen load
heavily on Fuel Economy.

\section{Discussion and challenges}
\label{sec:discussion}

In this section we put forward some 
points for debate, and discuss open
challenges for the development of theory 
and methods to deal with cellwise 
outliers.

\subsection{Interpretation of outliers
            as cellwise or casewise}
\label{sec:interpr}

The notions of casewise and cellwise
outliers raise some questions of
interpretation.
On the one hand, the terms casewise 
contamination and cellwise 
contamination have a clear meaning as 
generative models, i.e. as recipes 
for different ways to create outliers. 
Many researchers have used these models
to set up simulation studies.

On the other hand, when faced with
incoming data, inferring whether an 
unusual row is a casewise outlier 
(i.e. generated by a different 
mechanism from the majority of the 
cases) or a mostly clean row with 
some outlying cells, is often an
unidentifiable problem. 
As an illustration, consider a toy 
example from \cite{cellHandler}. 
Assume the clean data is i.i.d. from
the 4-dimensional standard Gaussian 
distribution $F = N(\bzero, \bI)$
and consider the row $(10, 0, 0, 0)$.
From the casewise point of view, this 
case is ``equivalent'' to the row
$(\sqrt{50}, \sqrt{50}, 0, 0)$, and 
to the row $(5, 5, 5, 5)$, as both 
can be obtained by an orthogonal 
transformation of the data, under 
which the distribution 
$N(\bzero, \bI)$ remains the same.
But from the cellwise point of view, 
one would conclude that the first row
has one cellwise outlier, the second 
has 2 and the last has 4.
Therefore, the conclusions depend on
whether we adopt the casewise or the
cellwise paradigm.

While a data-based determination of
casewise versus cellwise outliers may 
not be feasible, it may be not 
necessary to make such a distinction 
from the viewpoint of estimation. In his 
Section 4.1, \cite{danilov2010robust} 
states that ``Detection methods for 
cellwise contamination can be used to
disarm casewise outliers if the 
covariance structure is known well enough 
to identify them. If we remove several 
cells from a casewise outlier so that the 
remaining sub-vector is not outlying any
more, its influence on the covariance 
estimate is significantly reduced and can 
be tolerated.'' 
This viewpoint seems sensible to us. 
But in order to efficiently disarm 
adversarial casewise outliers, one would 
need a rather accurate fit (covariance 
estimate) when looking for the most 
harmful cells. 
Pairwise correlation methods are likely 
to suffer substantially under adversarial 
casewise outliers, and thus may have to 
be combined by some method targeting 
casewise outliers as well, like 
2SGS. It does seem like the ``disarming'' 
approach can work.
When applying the cellwise MCD estimator
to all the benchmark datasets in the 
\texttt{robustbase} R package 
\citep{robustbase} that were 
previously analyzed by casewise
robust methods, the fits turned out
to be quite similar \citep{cellMCD}.

An alternative view is that, in the 
cellwise world, we can consider 
something a casewise outlier when too 
many cells are flagged or need to be 
changed in order to bring the case 
into the fold. 
This idea has been used as an 
algorithmic tool in the 2SGS method,
where cases get downweighted as a
whole in the GSE step, and in the
MacroPCA method.
As a way to interpret the results it
has also been implemented in \textbf{DDC},
which has a rule to flag rows that
contain much cellwise outlyingness.
There it is a kind of 
post-processing step. 
Nevertheless one should remain 
conscious of the fact that this 
viewpoint does not quite match the 
generative models of rowwise and 
cellwise outliers, unless one
assumes a model that generates both.
Several authors have indeed 
generated data with both types of
outliers combined.

\subsection{Insights from sparsity}

Robustness considerations have recently 
received increasing attention in the 
computer \mbox{science} and electrical
engineering communities. 
One particularly popular approach to
casewise robustness is to take an 
existing objective function and then to 
add a penalty term about a parameter 
vector which encodes outliers. 
The estimated parameter vector thereby
becomes sparse, corresponding to a low 
number of flagged outliers.
To the best of our knowledge, the first 
reference to this idea is 
\cite{sardy2001robust} who drew a 
parallel between using a Huber loss 
function in regression and penalizing 
a shift parameter $\bdelta = 
(\delta_1, \ldots, \delta_n)$ with an 
$L^1$ penalty. 
In particular, they showed that 
minimizing
\begin{equation}
	\sum_{i=1}^{n}\left(y_i - \bbeta 
	\bx_i - \delta_i \right)^2 + 
	\lambda \sum_{i=1}^{n}{|\delta_i|}
\end{equation}
over $\bbeta$ and $\bdelta$ gives the
same $\bhbeta$ as minimizing 
\begin{equation}
  \sum_{i=1}^{n}
	\rho_{\mbox{\tiny Huber}, \lambda}
	\left(y_i - \bbeta \bx_i \right)\;.
\end{equation}
This approach was also followed by 
\cite{mccann2007robust},
\cite{laska2009exact} and
\cite{li2013compressed}, and was further 
developed by \cite{she2011outlier} who 
formulated general connections with 
M-estimation in robust statistics. 
They showed that non-convex penalties 
on $\bdelta$ lead to more robust 
estimators, in the same way that  
M-estimators with non-convex loss 
functions can be more 
robust than M-estimators with 
convex loss functions. 
In spite of this fact, still 
more research gets published for 
convex than for non-convex penalties 
and loss functions, because convex 
functions are easier for proving 
theorems and writing algorithms.

A similar approach can be applied
to cellwise outliers. 
Instead of an $n$-variate vector 
of shifts $\bdelta$, one
can use an $n \times d$ matrix 
$\bDelta$ of shifts to describe
cellwise outliers, and penalize
$\bDelta$ to make it sparse.
In the context of low-rank matrix
completion, this was done by
\cite{zhou2010stable} and
\cite{candes2011robust}. 
The context is not the same as ours
because of two main reasons.
First, their approach is presented
for PCA but it is equivariant
to transposing the data matrix,
whereas in the robust statistical
literature on PCA there can be 
rowwise outliers, e.g. fully 
contaminated rows, but fully 
contaminated columns are not 
allowed because they would make 
the loadings unidentifiable.
A more fundamental difference is 
that their approach does not aim
for robustness against adversarial 
contamination but instead assumes 
that the contaminated cells occur 
at random positions according to
the uniform distribution, and
independently of each other and
the clean data.
This assumption is why they
have no need for a non-convex
penalty, in fact they use the
$L^1$ penalty.
Their method has received a lot of 
attention and performs well under
their assumptions if the outliers 
do not have large orthogonal 
distances to the true low-rank 
subspace 
\citep{she2016robust}. 

For regression, one could use an
analogous approach by the 
minimization
\begin{equation} \label{eq:Delta} 
  \argmin_{\bbeta,\bDelta} \left(
	\sum_{i=1}^{n}\left(y_i - \bbeta 
	(\bx_i - \bDelta_{i,\cdot})\right)^2
	+ \lambda\mathcal{P}(\bDelta)\right)
\end{equation}
where $\mathcal{P}(\cdot)$ denotes 
a penalty on the cells in $\bDelta$. 
\cite{zhu2011sparsity} have proposed 
a similar idea with the convex 
Frobenius norm as penalty, i.e. 
$\mathcal{P}(\cdot) = ||\cdot||_F$. 
This is designed to work when the 
matrix of shifts is dense but bounded, 
but less suitable for a sparse set of 
potentially gross cellwise outliers. 
\cite{chen2013robust} go further in
this direction.
Like \cite{she2011outlier} they 
conclude that approaches
based on convex optimization cannot
be robust to casewise or cellwise
outliers in the covariates.
Instead they work with trimmed
inner products.

For the estimation of a covariance 
matrix one also could try an 
approach similar to~\eqref{eq:Delta}.
It would be natural to minimize the
negative Gaussian loglikelihood of
the shifted data as in
\begin{equation}
  \argmin_{\bmu,\bSigma,\bDelta} 
	\left(
	\sum_{i=1}^{n}{\left(\log|\bSigma| 
	+ (\bx_i - \bDelta_{i,\cdot} 
	- \bmu)^{\btop} \bSigma^{-1}
	(\bx_i - \bDelta_{i,\cdot}-\bmu)
	\right)} +	
	\lambda \mathcal{P}(\bDelta)
	\right)
\end{equation}
where $\mathcal{P}(\cdot)$ again 
denotes a penalty on the elements 
of $\bDelta$. 
While this optimization can be 
carried out quite efficiently, 
we have found that it leads to a 
bias which shrinks the estimated 
covariance matrix. 
The cellwise MCD of \cite{cellMCD} is 
a further development of this 
approach.
Rather than penalizing a matrix 
$\bDelta$ of shifts, it penalizes 
an $n \times d$ matrix $\bW$ with 
zeroes and ones that indicate which 
cells to include in the estimation. 
The objective function then becomes 
the observed likelihood of the 
incomplete dataset formed by the
included cells. 
This avoids the bias in the 
estimated covariance matrix.

\subsection{Models for cellwise outliers}

\cite{alqallaf2009} provided a
formalization of the idea of cellwise 
outliers, aimed at replacing the 
casewise Tukey-Huber contamination
model (THCM) of~\eqref{eq:thcm}. 
More precisely, they assumed that we 
observe a random variable generated by
\begin{equation}\label{eq:icm}
  \bX = (\bI - \bB) \bY + \bB\bZ 
\end{equation}
where $\bB$ is a diagonal matrix whose
diagonal elements can only take on the
values zero or one. To be precise, they
assume $\bB = \diag(B_1,\ldots,B_d)$ 
with Bernoulli distributed $B_j \sim 
\mbox{Bernoulli}(\varepsilon)$. 
\cite{alqallaf2009} refer to this 
as the 
{\it independent contamination model} 
(ICM) but one has to be a little 
careful with this name, because the
word {\it independent} does not refer 
to statistical independence, it merely
means that the diagonal elements of
$\bB$ are {\it allowed} to differ from 
each other. 
As \cite{alqallaf2009} point out, the 
THCM can be recovered from 
\eqref{eq:icm} by imposing the
equality $B_1 = \cdots = B_d$, and
then the $B_1, \ldots, B_d$ are not 
statistically independent.
At the other extreme is the situation
where $B_1, \ldots, B_d$ are, in fact,
statistically independent.
If on top of that we also assume that
$\bB$ is independent of $\bY$ and 
$\bZ$, we arrive at what they call
the \textit{fully independent 
contamination model} (FICM). 
Under the FICM each cell has a 
probability of $\varepsilon$ to 
be contaminated, and whether a cell 
gets contaminated does not 
depend on the value of the clean 
$\bY$ or the contamination $\bZ$.

Despite the fact that developing methods 
that work under the FICM model is already
hard, we would like to point out that
FICM is a restrictive model that could 
be extended substantially, as already
mentioned by \cite{alqallaf2009}. 
Recall that a random vector $\bX$ 
following the FICM can be written as
in~\eqref{eq:icm} where $\bY$ follows 
some clean distribution $F$ and $\bZ$ 
follows an unspecified distribution $H$. 
Importantly, $B_1,\ldots,B_d,\bY,\bZ$ 
are assumed independent. 
One could argue that this independence 
assumption is rather strong. 
Furthermore, in most empirical studies 
so far, the outlying cells have been 
generated with the additional 
assumption that also the components 
of $\bZ$ are independent of each 
other, and often even constant. 
This leads to cells that are marginally 
outlying, and can thus in principle be 
detected by looking at the marginal 
distributions of the variables.

Under the casewise THCM model 
of~\eqref{eq:thcm}
it is typically assumed that $B$ is 
independent of $\bY$, i.e. whether a case 
is sampled from the contaminated 
distribution does not depend on the 
value of the clean data. 
In general no assumptions on $H$ are made,
but in simulations $H$ sometimes depends 
on the {\it distribution} $F$ of the 
clean data. 
For instance, for covariance estimation, 
one often generates outlying cases in the 
direction of the last eigenvector of 
the true covariance matrix of $F$, which
is typically the most adversarial choice.

The FICM model could be extended in an
analogous way. As in THCM, we do not 
want $\bB$ to depend on $\bY$.
But the values of $\bZ$ could depend
on $\bB$ and on the distribution $F$ 
of the variable $\bY$.
This is what happens in the structural 
(adversarial) cellwise outlier 
generation in 
\cite{cellHandler,cellMCD}.
There the diagonal entries of $\bB$ 
are independently drawn from
$\mbox{Bernoulli}(\varepsilon)$.
Next, the contaminated cells are 
computed from the last eigenvector 
of the covariance matrix of the 
distribution of $\bY$ restricted to 
the variables with $B_j = 1$.
This uses the distribution of $\bY$,
but not the value of $\bY$ itself.

These dependencies can be 
characterized by the directed acyclic 
graph
\begin{center}
\includegraphics[width=0.5\columnwidth]
  {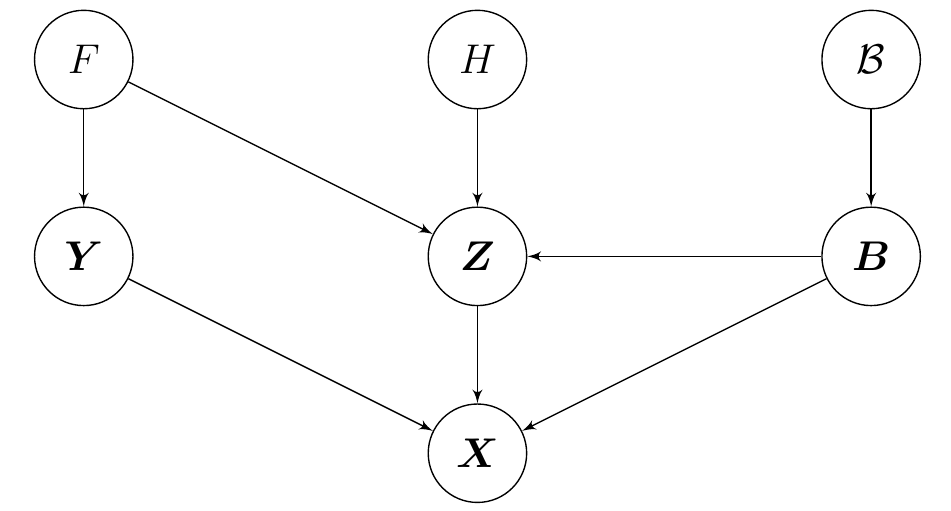}
\end{center}
where $\mathcal{B}$ is the multivariate
distribution of $\bB$. The final 
observed $\bX$ is at the bottom.
Note that there is an arrow from the
distribution $F$ to $\bZ$, but not from
the value $\bY$ to $\bZ$.
We feel that this model comes closest 
to the philosophy that has been 
dominant in the THCM setting.
This model allows for much more challenging 
cellwise outliers than have been 
predominantly considered so far. 
On the other hand, it is a valid question 
whether such outliers are actually common 
in practice. In order to answer this 
question we need appropriate methods to 
detect them.

\subsection{Cellwise outliers and heavy tails}

We feel it is important to stress the 
difference between the contamination 
model~\eqref{eq:icm} and the concept 
of heavy tails. 
Heavy-tailed distributions are typically 
defined as probability distributions 
whose tails are not exponentially 
bounded, such as the Cauchy 
distribution.
Methods for heavy-tailed data have 
received a lot of attention recently, 
e.g. by \cite{fan2021shrinkage} 
who avoid the ubiquitous sub-Gaussian 
restriction. 
But note that in this context there 
is no notion of adversarial 
contamination or structured outliers.

While heavy-tailed distributions may 
often generate samples that appear 
as though they could have 
been generated from~\eqref{eq:icm}, 
there is a crucial difference. 
If the marginal distributions of an 
observed random vector follow a 
heavy-tailed distribution, the 
cells that seem outlying still carry 
information about the underlying 
distribution. 
In contrast, under the 
model~\eqref{eq:icm} we don't want 
to assume anything about the 
distribution of $\bZ$, so we cannot 
use $\bZ$ to obtain information on 
the characteristics of the 
distribution of $\bY$, which is the 
one of interest.
Therefore cells generated by $\bZ$
do not carry information on $\bY$.
In the cellwise contamination model
we should thus strive to eliminate 
the effect of $\bZ$. 
If we were told which cells were 
generated from $\bZ$, then our 
best bet would be to discard them 
and to continue with the 
remaining cells. 

As a simple example, consider a 
univariate dataset with a single 
far out value. 
If we knew for a fact that the data 
came from a Cauchy distribution, 
then we should keep this unusual 
value and calculate the MLE for the 
Cauchy distribution. 
If, instead, we knew that the data 
came from a Gaussian distribution 
except for the suspicious value that 
came from some other distribution, 
we should discard it and perform 
the Gaussian MLE on the remaining 
values.

In spite of the two different models 
and objectives, there may still be merit 
in combining ideas from both worlds. 
In general, we expect that estimators 
for heavy-tailed data can still perform 
reasonably well under cellwise 
contamination. 
This may at least give them a role as 
initial estimators for iterative 
algorithms intended for data that
may contain cellwise outliers.

\section{Conclusion} \label{sec:conc}

Cellwise outliers provide a fascinating 
but tough challenge for statistical 
data analysis. 
Here we have reviewed some methodology 
that has been developed for detecting 
cellwise outliers and estimating location 
and covariance, and for carrying out a 
number of other tasks such as regression 
and PCA. 
From Tables~\ref{tab:flagging}
to~\ref{tab:pca} we conclude that
computation time is not an 
obstacle to the use of cellwise
robust methods.
We showed that the cellwise breakdown 
values of estimators of location, 
covariance and regression are rather 
low when holding on to some intuitive 
notions, meaning that one has to search 
in new directions.
We next proposed a cellwise robust 
approach for correspondence analysis.
Finally, we discussed various issues
that serve as food for thought and may 
inspire debate and further research.\\

\noindent{\bf Software availability.}
The code for the cellwise robust 
correspondence analysis method described 
in section~\ref{sec:CA} is available in 
the \textsf{R} package \texttt{cellWise} 
\citep{R:cellWise}, with a vignette
\texttt{Correspondence\_analysis\_examples}
reproducing Figures~\ref{fig:ddc}
to~\ref{fig:brands_biplot}.\\

\noindent {\bf Acknowledgment.}
Thanks go to Mia Hubert for interesting
discussions.

\section*{Appendix}

Here we give the proofs of the breakdown
results in section~\ref{sec:bdv}.

\begin{proof}[Proof of Proposition 
              \ref{prop:loc}]
Consider a hyperplane $H$ given by
the equation $\bx^{\btop} \bone_d = c$ 
where $\bone_d$ is the $d \times 1$ 
vector of ones, and $c$ is chosen large 
enough so that all data points $\bx_i$
satisfy $\bx_i^{\btop} \bone_d < c$.
So $\bX$ lies on one side of $H$,
and $H$ is not parallel to any 
coordinate axis.
In each row of $\bX$ we can then
replace a single cell such that all of
the resulting points lie on $H$, 
yielding the contaminated dataset 
$\bX^m$\,.
We can do this by replacing no more than
$\lceil n/d \rceil$ cells in each 
variable, which is a fraction
$\lceil n/d \rceil/n$ of its $n$ cells.
From property~\eqref{eq:locEFP} 
it then follows that $\bhmu(\bX^m)$
lies on $H$ too. Since we can choose
the value $c$ arbitrarily large, 
$\bhmu(\bX^m)$ breaks down.
\end{proof}

\begin{proof}[Proof of Corollary
              \ref{prop:loc2}]
Consider a dataset 
$\{\bx_i;\; i=1,\ldots,n\}$ 
that lies on a hyperplane $H$. Then
the $\bx'_i = \bx_i - \bx_1$ lie 
on the shifted hyperplane $H_0$
which passes through the origin.
Now take an orthogonal matrix $\bU$
which maps $H_0$ to the hyperplane 
spanned by the first $d$ orthonormal
basis vectors.
This maps the data points to
$\btx_i = \bU\bx'_i$.
Now take the orthogonal matrix
$\bV = \mbox{diag}(1,\ldots,1,-1)$.
By orthogonal equivariance we then
obtain $\bhmu(\{\bV\btx_i\}) =
\bV(\bhmu(\{\btx_i\})).$
But $\{\bV\btx_i\} = \{\btx_i\}$
hence $\bhmu(\{\btx_i\}) = 
\bV(\bhmu(\{\btx_i\}))$ so its last
coordinate must be zero.
Therefore $\bhmu(\{\bx_i\}) =
\bU^{\btop}\bhmu(\{\btx_i\}) + \bx_1$
lies on $H$, which 
proves~\eqref{eq:locEFP}.  
\end{proof}

\begin{proof}[Proof of Proposition 
              \ref{prop:impl}]
When $n \leqslant d$ all points of $\bX$
lie on a lower-dimensional affine 
subspace, so~\eqref{eq:covEFP} implies that
$\lambda_d(\bhSigma(\bX)) = 0$ so
$m=0$ already suffices for breakdown.
From here on we can assume $n > d$.
Pick one row of the dataset, and
consider an affine hyperplane that is
not parallel to any coordinate axis.
In the remaining $n-1$ rows we can then
replace a single cell such that all of
the resulting points lie on the same 
hyperplane, yielding the contaminated
dataset $\bX^m$\,.
We can do this by replacing no more than
$\lceil (n-1)/d \rceil$ cells in each 
variable, which is a fraction
$\lceil (n-1)/d \rceil/n$ of its $n$ 
cells. From the exact fit 
property~\eqref{eq:covEFP} 
it then follows that 
$\lambda_d(\bhSigma(\bX^m))=0$, so we
obtain implosion.
\end{proof}

\begin{proof}[Proof of Corollary
              \ref{prop:impl2}]
Consider a dataset $\btX$ in a 
lower-dimensional subspace.
Let us shift and rotate $\btX$ so that
its image $\bY$ lies in the hyperplane 
through the origin $\bzero$ spanned by 
the first $d-1$ orthonormal basis 
vectors.
By affine equivariance of $\bhSigma$ 
it follows that $\bC := \bhSigma(\bY)$ 
has the same eigenvalues as 
$\bhSigma(\btX)$.
We now split up $C$ in blocks of sizes
$d-1$ and 1, and consider the following
matrix $\bA$:
$$\begin{array}{ll}
  \bC=\begin{bmatrix}
  \bC_{11} & \bC_{21}^{\btop}\\
  \bC_{21} & \bC_{22}
\end{bmatrix}
\;\;\mbox{ and }
& \bA=\begin{bmatrix}
  \bI_{d-1} & \bzero\\
  \bzero & 2
  \end{bmatrix}\;.
  \end{array} $$
Now linearly transform $\bY$ to 
$\btY = \bY \bA$. By the affine 
equivariance of $\bhSigma$ we obtain
$$\bhSigma(\widetilde{\bY}) = 
  \bA^{\btop} \bC \bA =
  \begin{bmatrix}
  \bC_{11} & 2\bC_{21}^{\btop}\\
  2\bC_{21} & 4\bC_{22}
\end{bmatrix}\;.$$
But the latter matrix must equal $\bC$ 
since by construction $\btY = \bY$.
Therefore $2\bC_{21} = \bC_{21}$ and 
$4\bC_{22} = \bC_{22}$ which is only 
possible when $\bC_{21} = \bzero$ and 
$\bC_{22} = 0$, so $\bC$ is singular.
\end{proof}

\begin{proof}[Proof of Proposition 
              \ref{prop:reg}]
Take the first case $\bz_1$ of the 
dataset, and consider any nonzero number 
$\beta_0$\,. Then compute $\alpha_0 =
y_1 - \beta_0 x_{11} - \ldots -
\beta_0 x_{1p}$\,.
Combine them in the vector $\bgamma_0 =
(\alpha_0, \beta_0, \ldots, 
\beta_0)^{\btop}$.
Case $1$ then lies on the hyperplane
$H$ with equation 
$y = \bx^{\btop} \bgamma_0$\;.
Note that $H$ is not parallel to any
of the coordinate axes.
In the remaining $n-1$ cases we can then
replace a single cell of $\bZ$ such that 
all of the resulting points lie on $H$, 
yielding the contaminated
dataset $\bZ^m$\,.
We can do this by replacing no more than
$\lceil (n-1)/(p+1) \rceil$ cells in each 
variable (including the response), which 
is a fraction 
$\lceil (n-1)/(p+1) \rceil/n$ 
of its $n$ cells. From~\eqref{eq:regEFP} 
it then follows that 
$\bhgamma(\bZ^m) = \bgamma_0$.
Since $\beta_0$ can be chosen arbitrarily
large, $\bhgamma(\bZ^m)$ breaks down.
\end{proof}

\begin{proof}[Proof of Corollary
              \ref{prop:reg2}]
Consider a dataset $(\bX,\by)$ for which 
$y_i = \bx_i^{\btop}\bgamma_0$ for some 
$\bgamma_0$\;, for all $i = 1,\ldots,n$.
Construct $\bty = \by - \bX\bgamma_0 =
\bzero$.
By regression equivariance,
$\bhgamma(\bX,\bty) = \bhgamma(\bX,\by)
- \bgamma_0$\;.
Using the fact that $\bty = 0\bty$
and scale equivariance with $c=0$ 
we get $\bhgamma(\bX,\bty) =
\bhgamma(\bX,0\bty) =
0\bhgamma(\bX,\bty) = \bzero$.
Therefore 
$\bhgamma(\bX,\by) = 
\bhgamma(\bX,\bty) + \bgamma_0 =
\bgamma_0$ which
proves~\eqref{eq:regEFP}.
\end{proof}


\end{document}